\documentclass[a4paper,UKenglish,cleveref, autoref, thm-restate]{lipics-v2021}
\usepackage{wrapfig}

\usepackage{graphicx}

\usepackage[linesnumbered,lined,boxed,commentsnumbered]{algorithm2e}
\usepackage{multirow}
\usepackage{tabularx}
\usepackage{tikz}
\usetikzlibrary{calc}

\usepackage[]{color}

\title{Online TSP with Known Locations} %TODO Please add

\author{Evripidis Bampis}{Sorbonne Universit\'e, CNRS, LIP6, F-75005 Paris, France}{evripidis.bampis@lip6.fr}{}{}

\author{Bruno Escoffier}{Sorbonne Universit\'e, CNRS, LIP6, F-75005 Paris, France \and Institut Universitaire de France, Paris, France}{bruno.escoffier@lip6.fr}{}{%(Optional) author-specific funding acknowledgements
}

\author{Niklas Hahn}{Sorbonne Universit\'e, CNRS, LIP6, F-75005 Paris, France}{niklas.hahn@lip6.fr}{}{}

\author{Michalis Xefteris}{Sorbonne Universit\'e, CNRS, LIP6, F-75005 Paris, France}{michail.xefteris@lip6.fr}{}{%(Optional) author-specific funding acknowledgements
}

\authorrunning{E. Bampis, B. Escoffier, N. Hahn and M. Xefteris} %TODO mandatory. First: Use abbreviated first/middle names. Second (only in severe cases): Use first author plus 'et al.'

\Copyright{ } %TODO mandatory, please use full first names. LIPIcs license is "CC-BY";  http://creativecommons.org/licenses/by/3.0/

 %TODO mandatory; please add comma-separated list of keywords

\category{} %optional, e.g. invited paper

\relatedversion{} %optional, e.g. full version hosted on arXiv, HAL, or other respository/website
\funding{This work was partially funded by the grant ANR-19-CE48-0016 from the French National Research Agency (ANR).}%optional, to capture a funding statement, which applies to all authors. Please enter author specific funding statements as fifth argument of the \author macro.

\acknowledgements{%I want to thank \dots
}%optional

\nolinenumbers %uncomment to disable line numbering

\newcommand{\opt}{\ensuremath{OPT}}
\newcommand{\alg}{\ensuremath{ALG}}
\newcommand{\ratio}{{\rho}}

\begin{document}

\maketitle

\begin{abstract}

In this paper, we consider the Online Traveling Salesperson Problem (OLTSP) where the locations of the requests are known in advance, but not their arrival times. We study both the open variant, in which the algorithm is not required to return to the origin when all the requests are served, as well as the closed variant, in which the algorithm has to return to the origin after serving all the requests.  Our aim is to measure the impact of the extra knowledge of the locations on the competitiveness of the problem. We present an online 3/2-competitive algorithm for the general case and a matching lower bound for both the open and the closed variant. Then, we focus on some interesting metric spaces (ring, star, semi-line), providing both lower bounds and polynomial time online algorithms for the problem.  
\end{abstract}
\section{Introduction}
In the classical Traveling Salesperson Problem (TSP), we are  given a set of locations in a metric space. The objective is to find a tour visiting all the locations minimizing the total traveled time assuming that the traveler moves in constant speed \cite{Lawler91}. Ausiello et al.\ \cite{AusielloFLST01} introduced the Online Traveling Salesperson Problem (OLTSP) where the input arrives over time, i.e., during the travel new requests (locations) appear that have to be visited by the algorithm. The time in which a request is communicated to the traveler (from now on we will refer to him/her as the server) is called release time (or release date).  Since then, a series of papers considered many versions of OLTSP in various metric spaces (general metric space \cite{AusielloFLST01}, the line \cite{AusielloFLST01, BjeldeHDHLMSSS21,GouleakisLS22}, or the semi-line \cite{AusielloDLP04,BernardiniLMMSS22,BlomKPS01}). 
Several applications can be modeled as variants of  OLTSP, e.g. applications in logistics and robotics \cite{AscheuerGKR90, PsaraftisSMK90}. 

The performance of the proposed algorithms for OLTSP has been evaluated in the framework of competitive  analysis, by establishing upper and lower bounds of the competitive ratio which is defined as the maximum ratio between the cost of the online algorithm and the cost of an optimal offline algorithm over all input instances. However, it is admitted that the competitive analysis approach is sometimes very pessimistic as it gives a lot of power to  the adversary. Hence, many papers try  to limit the power of the adversary by  suggesting for instance the notion of a fair adversary \cite{BlomKPS01}, or by giving extra knowledge  and hence more power to the online algorithm by introducing the notion of $\Delta$-time-lookahead \cite{AllulliAL05}, where the online algorithm  knows all the requests that will arrive in the next $\Delta$ units of time, or the $k$-request-lookahead, where the online algorithm has access to the next $k$ requests \cite{AllulliAL05}. In the same vein, Jaillet and Wagner \cite{Jaillet} introduced the notion of disclosure dates, i.e., the dates at which the requests become known to the online algorithm ahead of their release dates. More recently, OLTSP has also been studied in the framework of {\em Learning-Augmented Algorithms} \cite{BernardiniLMMSS22,GouleakisLS22,HuWLCL22}. Here, we study another natural way of giving more power  to the online algorithm by considering that the location of each request is known in advance, but its release date arrives over time. The release date, in this context, is just the time after which a request can be served. This is the case for many applications where the delivery/collection locations are known (parcel collection from fixed storage facilities, cargo collection on a harbour etc.). Think for example of a courier that has to deliver packets to a fixed number of customers. These packets may have to be delivered in person, so the customers should be at home to receive them. Each customer can inform the courier through an app when he/she returns home and is ready to receive the packet. We refer to this problem as the online Traveling Salesperson Problem with known locations (OLTSP-L) and consider two variants: in the {\em closed} variant the server is required to return to the origin after serving all requests, while in the {\em open} one the server does not have to  return to the origin after serving the last  request. 
%In previous studies of OLTSP every request is characterized by a pair including its  location and its release time. However, in some applications the locations are known in advance, the only online characteristic being the release times of the requests.  
%Here, we are interested in the influence of this extra knowledge on the competitiveness of the OLTSP problem. 

\paragraph*{Previous results} 
 The offline version of the open variant of the problem in which the requests are known in advance can be solved in quadratic time when the metric space is a line  \cite{PsaraftisSMK90, BjeldeHDHLMSSS21}. In \cite{BjeldeHDHLMSSS21}, Bjelde et al.\ studied the offline closed variant providing a dynamic program based on the one of \cite{PsaraftisSMK90} that solves the problem in quadratic time.
For the online problem, Ausiello et al.\ \cite{AusielloFLST01} showed a lower bound of $2$ for the open version and a lower bound of $1.64$ for the closed version of OLTSP, even when the metric space is the line. They also proposed an optimal $2$-competitive algorithm for the closed version and a $2.5$-competitive algorithm for the open version of OLTSP in general metric spaces. For the line, they presented a $(7/3)$-competitive algorithm for the open case and a $1.75$-competitive algorithm for the closed one. More recently, Bjelde et al., in \cite{BjeldeHDHLMSSS21}, proposed a $1.64$-competitive algorithm for the close case (on the line) matching the lower bound of \cite{AusielloFLST01}. They also provided a lower bound of $2.04$ for the open case on the line, as well as an online algorithm matching this bound.  In \cite{BlomKPS01}, Blom et al.\ proposed a best possible $1.5$-competitive algorithm when the metric space is a semi-line. Chen et al., in \cite{ChuanDLW19}, presented lower and upper bounds of randomized algorithms for OLTSP on the line. 

In \cite{HuWLCL22}, Hu et al.\ introduced learning-augmented algorithms for OLTSP. They proposed three different prediction models. In the first model,  the number of requests is not known in advance and  each request is associated to a prediction for both its release time and its location. In the second model, they assume that the number of requests is given and that, as in the first model,  a prediction corresponds to both the release time and the location of the request. While being closer to our model, no direct comparison can be done between their results and ours. %This model is more related to ours, however their results cannot be compared to our results. 
In the third model,  the prediction is just the release time of the last request. In \cite{BernardiniLMMSS22}, Bernardini et al.\ focused also on learning-augmented algorithms and they introduced a new error measure appropriate for many online graph problems. For OLTSP, they assume that a prediction corresponds to both the release time and the location of a request.  They study OLTSP for metric spaces, but also the more special case of the semi-line.   Their results are not comparable to ours. In \cite{GouleakisLS22}, Gouleakis et al.\ studied a learning-augmented framework for OLTSP on the line. The authors define a prediction model in  which the predictions correspond to the locations of the requests. They establish lower bounds by assuming that the predictions are perfect, i.e.  the locations  are given,  and the adversary can only control the  release times of the requests.   They also proposed upper  bounds as a function of the value of the error. Let us  note that in the case where the error is equal to 0, their model coincides with ours (see Table~\ref{table:results} for a comparison with our results). 

\section{Our contribution}

In this work we study both the closed (closed OLTSP-L) and the open case of OLTSP-L (open OLTSP-L) and present several lower and upper bounds for the problem. In Table~\ref{table:results}, we give an overview of our results and the state of the art.

\begin{table}[ht]
\centering
    \scalebox{0.95}{
    \begin{tabular}{c||ccc|ccc}
        & \multicolumn{3}{c|}{Open OLTSP-L} & \multicolumn{3}{c}{Closed OLTSP-L} \\  & \multicolumn{1}{c|}{Lower Bound}                     & \multicolumn{2}{c|}{Upper Bound} & \multicolumn{1}{c|}{Lower Bound} & \multicolumn{2}{c}{Upper Bound}
              \\ \hline \hline
        Semi-line & \multicolumn{1}{c|}{\begin{tabular}[c]{@{}c@{}}4/3\\ Thm.~\ref{theo:semilineOpenLB}\end{tabular}}  & \multicolumn{2}{c|}{\begin{tabular}[c]{@{}c@{}}13/9*\\ Thm.~\ref{theo:semilineOpenCompetitiveRatio}\end{tabular}} & \multicolumn{1}{c|}{\textbf{1}} & \multicolumn{2}{c}{\begin{tabular}[c]{@{}c@{}}\textbf{1*}\\ Prop.~\ref{prop:semilineClosedUB}\end{tabular}}                           \\ \hline
        Line      & \multicolumn{1}{c|}{\begin{tabular}[c]{@{}c@{}}13/9\\ \cite[Thm.~4]{GouleakisLS22}\end{tabular}} & \begin{tabular}[c]{@{}c@{}}3/2\\Thm.~\ref{th:3/2general}\end{tabular} & \begin{tabular}[c]{@{}c@{}} 5/3*\\ \cite[Thm.~3]{GouleakisLS22}\end{tabular} & \multicolumn{1}{c|}{\begin{tabular}[c]{@{}c@{}}\textbf{3/2}\\ \cite[Thm.~2]{GouleakisLS22}\end{tabular}} & \multicolumn{2}{c}{\begin{tabular}[c]{@{}c@{}}\textbf{3/2*}\\ \cite[Thm.~1]{GouleakisLS22}\end{tabular}}                                    \\ \hline
        Star      & \multicolumn{1}{c|}{\begin{tabular}[c]{@{}c@{}}13/9\\\cite[Thm.~4]{GouleakisLS22}\end{tabular}} & \multicolumn{2}{c|}{\begin{tabular}[c]{@{}c@{}}3/2\\Thm.~\ref{th:3/2general}\end{tabular}}                                     & \multicolumn{1}{c|}{\begin{tabular}[c]{@{}c@{}}\textbf{3/2}\\\cite[Thm.~2]{GouleakisLS22}\end{tabular}} & \begin{tabular}[c]{@{}c@{}}\textbf{3/2}\\Thm.~\ref{th:3/2general}\end{tabular} & \begin{tabular}[c]{@{}c@{}}$(7/4+\epsilon)^*$ \\ Cor.~\ref{cor:StarClosedCompetitiveRatioPolyTime}\end{tabular}
                \\ \hline
        Ring      & \multicolumn{1}{c|}{\begin{tabular}[c]{@{}c@{}}\textbf{3/2}\\ Prop.~\ref{prop:ringOpenLB}\end{tabular}}  & \multicolumn{2}{c|}{\begin{tabular}[c]{@{}c@{}}\textbf{3/2}\\Thm.~\ref{th:3/2general}\end{tabular}}                                     & \multicolumn{1}{c|}{\begin{tabular}[c]{@{}c@{}}\textbf{3/2}\\\cite[Thm.~2]{GouleakisLS22}\end{tabular}} & \begin{tabular}[c]{@{}c@{}}\textbf{3/2}\\Thm.~\ref{th:3/2general}\end{tabular} & \begin{tabular}[c]{@{}c@{}}5/3*\\ Thm.~\ref{theo:ringClosedCompetitiveRatio2}\end{tabular}
                \\ \hline
        General   & \multicolumn{1}{c|}{\begin{tabular}[c]{@{}c@{}}\textbf{3/2}\\Prop.~\ref{prop:ringOpenLB} \end{tabular}}  & \multicolumn{2}{c|}{\begin{tabular}[c]{@{}c@{}}\textbf{3/2}\\ Thm.~\ref{th:3/2general}\end{tabular}} & \multicolumn{1}{c|}{\begin{tabular}[c]{@{}c@{}}\textbf{3/2}\\\cite[Thm.~2]{GouleakisLS22}\end{tabular}} & \multicolumn{2}{c}{\begin{tabular}[c]{@{}c@{}}\textbf{3/2}\\ Thm.~\ref{th:3/2general}\end{tabular}}                                      
    \end{tabular}
    }
    \caption{Results for the open and closed variants of OLTSP-L. Polynomial time algorithms are denoted by * and tight results in bold.}
    \label{table:results}
\end{table}

We first consider general metric spaces (dealt with in Section~\ref{GeneralMetrics}). For the closed version, a lower bound of 3/2 has been shown in~\cite{GouleakisLS22} (valid in the case of a line). We show that a lower bound of $3/2$ holds also for the open case (on rings). We then provide a $3/2$-competitive online algorithm for both variants, thus matching the lower bounds. Although our algorithm does not run in polynomial time, the online nature of the problem is a source of difficulty independent of its computational complexity. Thus, such algorithms are of interest even if their running time is not polynomially bounded. 

However, it is natural to also focus on polynomial time algorithms. We provide several bounds in specific metric spaces. In Section~\ref{Ring}, we focus on rings and present a polytime $5/3$-competitive algorithm for the closed OLTSP-L. Next, we give a polytime $(7/4+\epsilon)$-competitive algorithm, for any constant $\epsilon > 0$, in the closed case on stars (Section~\ref{Stars}). In Section~\ref{SemiLine}, we study the problem on the semi-line. We present a simple polytime $1$-competitive algorithm for the closed variant and, for the open case, we give a lower bound of $4/3$ and an upper bound (with a polytime algorithm) of $13/9$.

%In the rest of the paper, we focus on OLTSP with known locations (OLTSP-L). In Section~\ref{GeneralMetrics}, we study the open and the closed variant of the problem in general metric spaces. We show a lower bound of $3/2$ for open OLTSP-L on rings, and then provide a $3/2$-competitive online algorithm for both variants which is optimal. Although our algorithm does not run in polynomial time, the online nature of the problem is a source of difficulty independent of its computational complexity. Thus, such algorithms are of interest even if their running time is not polynomially bounded. However, it is natural to then focus on polynomial time algorithms for the problem in specific metric spaces. In Section~\ref{Ring}, we focus on rings and present a $5/3$-competitive algorithm for the closed OLTSP-L. Next, we give a $(7/4+\epsilon)$-competitive algorithm, for any constant $\epsilon > 0$, in the closed case on stars (Section~\ref{Stars}). In Section~\ref{SemiLine}, we study the problem on the semi-line. We present a simple $1$-competitive algorithm for the closed variant and, for the open case, we give a lower bound of $4/3$ and an upper bound of $13/9$.

To measure the gain of knowing the locations of the requests, we also provide some lower bounds on the case where the locations are unknown, and more precisely on the case where the number of locations is known but not the locations themselves. 
%Since the locations of the requests imply their number, one could also consider a model in which only the number of requests is known. For this weaker model, we provide lower bounds to better understand its power. 
In the open case of the problem, there is a lower bound of $2$ in~\cite{AusielloFLST01} on the line  that holds even when the number of requests is known. Hence the same lower bound holds for  the star and for the ring. We provide  a lower bound of $3/2$  (Proposition~\ref{prop:KnownNumberOpenSemiline}) for the semi-line.
For the closed case, we show a lower bound of 2 for the ring (Proposition~\ref{prop:KnownNumberClosedRing}) and the star (Proposition~\ref{prop:KnownNumberStar}). 
We also present a lower bound of  $4/3$ (Proposition~\ref{prop:KnownNumberClosedSemiline}) for the semi-line. For the line, it is easy to see (with a trivial modification in the proof) that the lower bound of $1.64$ on the line~\cite{AusielloFLST01} still holds in this model where we know the number of locations. These lower bounds show that knowing only the number of requests is not sufficient to get better competitive ratios in most cases, and thus it is meaningful to consider the more powerful model with known locations.

\section{Preliminaries}

The input of OLTSP-L consists of a metric space $M$ with a distinguished point $O$ (the origin), and a set $Q=\{q_1,...,q_n\}$ of $n$ requests. Every request $q_i$ is a pair ($t_i, p_i$), where $p_i$ is a point of $M$ (which is known at $t=0$) and $t_i \ge 0$ is a real number. The number $t_i$ represents the moment after which the request $q_i$ can be served (release time). A server located at the origin at time $0$, which can move with at most unit speed, must serve all the requests after their release times with the goal of minimizing the total completion time (makespan).

When we refer to general metric spaces or general metrics, we mean the wide class of metric spaces $M$ defined in \cite{AusielloFLST01}. This class contains all continuous metric spaces which have the property that the shortest path from $x\in M$ to $y \in M$ is continuous in $M$ and has length $d(x,y)$. %The release times for continuous metric spaces can be any non-negative real number. 
We note that our 3/2-competitive algorithm for general metric spaces also works for discrete metric spaces (where you do not continuously travel from one point to another, but travel in a discrete manner from point $x$ at time $t$ to point $y$ at time $t+d(x,y)$). %Additionally, $M$ contains discrete metric spaces that can be represented by an underlying graph with all edges having unit length. The vertices of the graph are the points of the metric space. In that case, time needs to be discretized, i.e., the release times are non-negative integers and the server determines its strategy (to remain idle or move to a neighboring point) at integer points in time.

For the rest of the paper, we denote the total completion time of an online algorithm $\alg$ by $|\alg|$ and that of an optimal (offline) solution $\opt$ by $|\opt|$. We recall that an algorithm $\alg$ is $r$-competitive if on all instances we have $|\alg|\leq r \cdot |\opt|$. We will often use $\ratio$ to denote $\frac{|\alg|}{|\opt|}$ and $t$ to quantify time.

\section{General metrics} \label{GeneralMetrics}

As mentioned earlier, \cite{GouleakisLS22} showed a lower bound of 3/2 for closed OLTSP-L, even in the case of a line. In this section, we first show that the same lower bound of 3/2 holds for open OLTSP-L (in the case of a ring). We then devise a 3/2-competitive algorithm, for general metrics, both in the open and in the closed cases, thus matching the lower bounds in both cases. 

Let us first show the lower bound for open OLTSP-L.

\begin{proposition}\label{prop:ringOpenLB}
    For any $\epsilon > 0$, there is no $(3/2-\epsilon)$-competitive algorithm for open OLTSP-L on the ring.
\end{proposition}
\begin{proof}
    Consider a ring with circumference 1 with 2 requests $A$ and $B$, with a distance of $1/3$ from $O$ and from each other as visualized in Figure~\ref{fig:circleLBRingOpen}. At time $t=1/3$, without loss of generality due to symmetry, we can assume that the algorithm is somewhere in the segment (arc of the ring) $[OA]$ (including both $O$ and $A$). Then, $B$ is released. The second request $A$ is released at time $2/3$. Hence, $\opt$ can visit $A$ and $B$ in time $2/3$, whereas the online algorithm cannot finish before $t=1$. Whether it serves $A$ or $B$ first, it cannot serve the first request before time $t=2/3$ and will have to go a distance of $1/3$ to serve the second request, as well.
\end{proof}

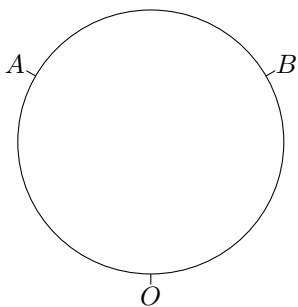
\begin{figure}[ht]
\begin{center}
\begin{tikzpicture}[scale=0.7]
	\draw (0,0) arc (270:-90:2.5);  %the syntax is: Start (start:end (in degrees):radius), where a degree of 0 (mod 360) is the right point, so our origin at the bottom is at 270 degrees
	%\draw[thick] (0,0) arc [start angle=270, delta angle=160, radius=2.5cm]; %different way of writing the above, where delta angle is now not the end, but the "length"
	\node[] at (0,-0.425) (0) {$O$};
    \draw (0,0) -- (0,-0.2);
	
	\node[] at ({2.95*cos(150) }, {2.5+2.95*sin(150) }) (0) {$A$};
 \draw ({2.5*cos(150)},{2.5+2.5*sin(150)}) -- ({2.7*cos(150)},{2.5+2.7*sin(150)});
\node[] at ({2.95*cos(30) }, {2.5+2.95*sin(30) }) (0) {$B$};
 \draw ({2.5*cos(30)},{2.5+2.5*sin(30)}) -- ({2.7*cos(30)},{2.5+2.7*sin(30)});
\end{tikzpicture}
   \caption{Lower bound instance for open OLTPS-L on the ring}
    \label{fig:circleLBRingOpen}
\end{center}
\end{figure}

Now, let us present the 3/2-competitive algorithm. Roughly speaking, the principle of the algorithm is: \begin{itemize}
    \item first, to wait (at the origin)  a well chosen amount of time $T$. This time $T$ depends both on the locations of the requests and on the time they are released. It is chosen so that (1) the optimal solution cannot have already visited (served the requests on) a large part of its tour (closed case) / path (open case) and (2) there is a tour/path for which a large part is fully revealed (which is a good tour to follow if not too long).
    \item then, to choose an order of serving requests that optimizes some criterion mixing the length of the corresponding tour/path and the fraction of it which is released at time $T$, and to follow this tour/path (starting at time $T$), waiting at requests if they are not released.
\end{itemize}

More formally, we consider Algorithm~\ref{algo:general}, which uses the following notation (see Example~\ref{ex:gen} for an example illustrating the notation and the execution of the algorithm). For a given order $\sigma_i$ on the requests (where $\sigma_i(1)$ denotes the first request in the order, $\sigma_i(2)$ the second request, \dots), we denote:
\begin{itemize} \item by $\ell_i$ the length of the tour/path associated to $\sigma_i$ (starting at $O$), i.e., $\ell_i=d(O,\sigma_i(1))+\sum_{j=1}^{n-1}d(\sigma_i(j),\sigma_i(j+1))$ in the open case, $\ell_i=d(O,\sigma_i(1))+\sum_{j=1}^{n-1}d(\sigma_i(j),\sigma_i(j+1))+d(\sigma_i(n),O)$ in the closed case; \item by $\alpha_i(t)$ the fraction of the tour/path associated to $\sigma_i$, starting at $O$, which is fully released at time $t$. More formally, if requests $\sigma_i(1),\dots,\sigma_i(k-1)$ are released at $t$ but $\sigma_i(k)$ is not, then the tour/path is fully released up to $\sigma_i(k)$, and $\alpha_i(t)=\frac{d(O,\sigma_i(1))+\sum_{j=1}^{k-1} d(\sigma_i(j),\sigma_i(j+1))}{\ell_i}$.
%\michalis{Change the notation of requests? We already have tour/order as $\sigma$.}
%\niklas{Also/similarly: Shouldn't it be $d(\sigma_i, \sigma_{i+1})$ rather than $d(p_i, p_{i+1})$ in the sum? I'm fine with abusing notation (but then we should clarify it somewhere) -- I think something like $d(p_{\sigma(i)},p_{\sigma(i+1)})$ (although it might be the correct way) does not look great and might be harder to grasp, some short notation would be nice.}
\end{itemize}

\begin{algorithm}[ht]
	\KwIn{Offline: Request locations $p_1, \dots, p_n$\\
	\phantom{\textbf{Input:} }Online: Release times $t_1, \dots, t_n$}
	Let $\sigma_1,\dots,\sigma_{n!}$ be the orders of requests.\\
	
	For all $i\leq n!$, compute $\ell_i$ the length of the tour/path associated to $\sigma_i$.\\
	
	Wait at $O$ until time $T$ defined as the first  time $t$ such that there exists an order $\sigma_{i_0}$ with (1) $t\geq \ell_{i_0}/2$ and (2) $\alpha_{i_0}(t)\geq 1/2$.\\
	
	At time $T$:
	\begin{itemize}\item Compute an order $\sigma_{i_1}$ which minimizes, over all orders $\sigma_i$ ($1\leq i\leq n!$), %(not only the ones fulfilling the previous conditions),
	$(1-\beta_i)\ell_i$,\\
    where  $\beta_i=\min\{\alpha_{i}(T),1/2\}$.

	\item Follow (starting at time $T$) the tour/path associated to $\sigma_{i_1}$, by serving the \\
    requests in this order, waiting at a request location if this request is not released.
	\end{itemize}%\\
    \caption{Algorithm for closed and open OLTSP-L }
	\label{algo:general}
\end{algorithm}

\begin{example}\label{ex:gen}
Let us consider the example of Figure~\ref{fig:gen} with three requests $q_1,q_2,q_3$, released respectively at time $2$, 6 and 8. We focus on the closed case.
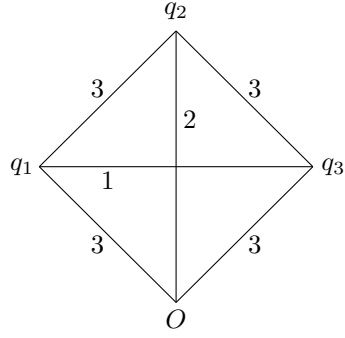
\begin{figure}[ht]
\begin{center}
\begin{tikzpicture}[scale=0.9]
	\node[] at (0,-0.25) (0) {$O$};
    \draw (0,0) -- (2,2);
	\node[] at (2.3,2) (0) {$q_3$};
    \draw (2,2) -- (0,4);
	\node[] at (0,4.3) (0) {$q_2$};
    \draw (-2,2) -- (0,4);
	\node[] at (-2.25,2) (0) {$q_1$};
    \draw (0,0) -- (-2,2);
    \draw (-2,2) -- (2,2);
    \draw (0,0) -- (0,4);    
	\node[] at (1.15,3.15) (0) {$3$};
	\node[] at (-1.15,3.15) (0) {$3$};
	\node[] at (1.15,0.85) (0) {$3$};
	\node[] at (-1.15,0.85) (0) {$3$};
	\node[] at (0.2,2.7) (0) {$2$};
	\node[] at (-1,1.8) (0) {$1$};

%	\node[] at ({2.85*cos(150) }, {2.5+2.85*sin(150) }) (0) {$A$};
% \draw ({2.5*cos(150)},{2.5+2.5*sin(150)}) -- ({2.7*cos(150)},{2.5+2.7*sin(150)});
%\node[] at ({2.85*cos(30) }, {2.5+2.85*sin(30) }) (0) {$B$};
% \draw ({2.5*cos(30)},{2.5+2.5*sin(30)}) -- ({2.7*cos(30)},{2.5+2.7*sin(30)});
\end{tikzpicture}
   \caption{An example with 3 requests. Weights correspond to distances.}
    \label{fig:gen}
\end{center}
\end{figure}

Let us consider the order $\sigma_0=(q_1,q_2,q_3)$, with $\ell_0=12$ (closed case). Then for $0\leq t<2$ we have $\alpha_0(t)=1/4$, for $2\leq t<6$ we have $\alpha_0(t)=1/2$, for $6\leq t<8$ we have $\alpha_0(t)=3/4$, and for $t\geq 8$ we have $\alpha_0(t)=1$.

Now, let us look at the algorithm, and first the determination of $T$ (line 3 of the algorithm). At time $6$, we have $6\geq \ell_0/2$ and $\alpha_0(6)\geq 1/2$, so $T\leq 6$. At any time $t<6$: on the one hand $t<\ell_0/2$, and on the other hand neither $q_2$ nor $q_3$ is released, and one can see that no tour $\sigma_i\neq \sigma_0$ has $\alpha_i(t)\geq 1/2$. This means that $T=6$, and that $i_0=0$. 

So the algorithm starts moving at $T=6$. The order minimizing $(1-\beta_i)\ell_i$ is $\sigma_{i_1}=(q_2,q_1,q_3)$ (with $\ell_{i_1}=9$, $\alpha_{i_1}(T)=6/9$, so $\beta_{i_1}=1/2$ and $(1-\beta_{i_1})\ell_1=9/2$), so the algorithm will follow this tour, serving $q_2$ at $6+1=7$, $q_1$ at 10, $q_3$ at 12 and being back in $O$ at 15.
\end{example}

\begin{theorem}\label{th:3/2general}
Algorithm~\ref{algo:general} is 3/2-competitive both for  closed and open OLTSP-L.
\end{theorem}
\begin{proof}

    At time $T$, there exists $\sigma_{i_0}$ with $T\geq \ell_{i_0}/2$ and $\alpha_{i_0}(T)\geq 1/2$. Then $\beta_{i_0}=1/2$, and $(1-\beta_{i_0})\ell_{i_0}=\ell_{i_0}/2\leq T$. By definition of  $\sigma_{i_1}$, we have $(1-\beta_{i_1})\ell_{i_1}\leq (1-\beta_{i_0})\ell_{i_0}$. So we get
    \begin{equation}\label{eq:gen1}
     (1-\beta_{i_1})\ell_{i_1}\leq T\enspace.
    \end{equation}
    
    Now, let us consider an optimal solution $\opt$, and denote by $\sigma_{i^*}$ the order in which $\opt$ serves the requests. We consider w.l.o.g.\ that $\opt$ follows the tour/path $\sigma_{i^*}$, waiting only at requests' positions\footnote{Indeed, if $\opt$ does not do this we can easily transform it into another optimal solution (with the same order of serving requests) that acts like this.}: it goes from $O$ to (the position of) $\sigma_{i^*}(1)$ in time $d(O,\sigma_{i^*}(1))$, waits at (the position of) $\sigma_{i^*}(1)$ if the request is not released, then from $\sigma_{i^*}(1)$ to $\sigma_{i^*}(2)$ in time $d(\sigma_{i^*}(1),\sigma_{i^*}(2))$, \dots 
    \begin{itemize}
        \item If $\alpha_{i^*}(T)\leq 1/2$, then $\beta_{i^*}=\alpha_{i^*}(T)$, and  $|\opt|\geq T+(1-\alpha_{i^*}(T))\ell_{i^*}=T+(1-\beta_{i^*})\ell_{i^*}$.
        \item If $\alpha_{i^*}(T) > 1/2$, then $\beta_{i^*}=1/2$. We look at the position of $\opt$ at time $T$ in the tour/path associated to $\sigma_{i^*}$. Suppose that it is (strictly) on the second half of this tour/path. Then, at $T-\epsilon$ (for a sufficiently small $\epsilon>0$), it was already on the second part. But then $T-\epsilon\geq \ell_{i^*}/2$ (as $\opt$ has already visited half of the tour/path), and $\alpha_{i^*}(T-\epsilon)\geq 1/2$ (for the same reason). Then $T-\epsilon$ would fulfill the two conditions in Line 3 of the algorithm, a contradiction with the definition of $T$. Consequently,  at $T$, $\opt$ is in the first half of its tour/path. So $|\opt|\geq T+\ell_{i^*}/2=T+(1-\beta_{i^*})\ell_{i^*}$. 
    \end{itemize}
    Then, in both cases we have $|\opt|\geq T+(1-\beta_{i^*})\ell_{i^*}$. By definition of $\sigma_{i_1}$ we deduce that
    \begin{equation}\label{eq:gen2}
     T+(1-\beta_{i_1})\ell_{i_1} \leq |\opt|\enspace.
    \end{equation}
    
    Now we look at the value $|\alg|$ of the solution $\alg$ output by the algorithm. We distinguish two cases:
    \begin{itemize}
        \item Suppose that $\alg$ never waits after $T$. Then $|\alg|=T+\ell_{i_1}$. As by definition $\beta_{i_1}\leq 1/2$, 
        $|\alg|\leq T+2(1-\beta_{i_1})\ell_{i_1}$. But adding Equations~(\ref{eq:gen1}) and~(\ref{eq:gen2})
        with coefficients $1/2$ and $3/2$, we get
        $T+2(1-\beta_{i_1})\ell_{i_1}\leq 3|\opt|/2$. 
        Hence, $|\alg|\leq 3|\opt|/2$.  
        \item Suppose now that $\alg$ waits after $T$ for some request to be released. Let $t^*$ be the last time $\alg$ waits. As a fraction $\alpha_{i_1}(T)$ of $\sigma_{i_1}$ is completely released at $T$ (i.e., when $\alg$ starts), $\alg$ has distance at most  $(1-\alpha_{i_1}(T))\ell_{i_1}$ to perform after $t^*$. So
    \begin{equation}\label{eq:gen3}
     |\alg|\leq t^*+(1-\alpha_{i_1}(T))\ell_{i_1} \leq t^*+(1-\beta_{i_1})\ell_{i_1}\enspace,
    \end{equation}
    where we use the fact that, by definition,  $\beta_{i_1}\leq \alpha_{i_1}(T)$.
    We have $t^*\leq |\opt|$, as a request is released at $t^*$. Adding Equations~(\ref{eq:gen1}) and~(\ref{eq:gen2}) gives $2(1-\beta_{i_1})\ell_{i_1}\leq |\opt|$. Putting these two inequalities in Equation~(\ref{eq:gen3}) gives $|\alg|\leq 3|\opt|/2$. 
    \end{itemize}
\end{proof}

Note that Algorithm~\ref{algo:general} also solves the more general Online Asymmetric Traveling Salesperson Problem with known locations and achieves again a ratio of $3/2$ for both variants. In the asymmetric version of TSP, the distance $d(a,b)$ from one point $a$ to another point $b$ in a given space can be different from the inverse distance $d(b,a)$. For the problem without the knowledge of the locations, there is an optimal $2.62$-competitive algorithm for the closed variant, and it has been proved that there is no constant competitive algorithm for the open variant~\cite{AUSIELLO2008290}.

\section{Ring} \label{Ring}

In this section, we discuss the problem on the metric space induced by the border of a ring. Without loss of generality, we can assume that the ring is a circle with a circumference of 1. The distance $d(x,y)$ between points $x,y$ on the border of the circle is the smaller arc length of the segment between $x$ and $y$ (or $y$ and $x$). We denote the location of points on the ring by their clockwise distance from the origin. 
Further, we assume that the requests are ordered such that $p_i$ is located next to $p_{i+1}$ for all $i \in [n-1]$ and $p_1$ and $p_n$ are the closest requests to the origin on their respective sides. 
%\bruno{If needed: We also assume that $p_1$ and $p_n$ are located at $O$ (with clockwise distance 0 and 1, to be consistent with the ordering) and released at $t=0$ (this does not modify the problem in both versions, as they are anyway immediately served in any solution).}

%\subsection{Closed case}
 In the closed case, for the classic online TSP, a lower bound of 2 on the competitive ratio for this metric space was shown in \cite{AusielloFLST01}. Note that the algorithm does not know the number of requests to be released,  but even an algorithm with this information would not be able to improve upon the competitive ratio of 2.
 %\niklas{Instead of re-releasing (only) requests that have already been served at the crucial time, the adversary could release a fixed number of requests which is (weakly) higher than the number of requests already visited. This way, the instance would always have a known number of requests. Should this also be included? (In an appendix?)}\bruno{Yes I think it should be included. In appendix is a good idea at this point.} 
More formally, we have the following statement (a formal proof can be found in Appendix~\ref{app:KnownNumberClosedRing}).
%We include a formal proof of this statement in Appendix~\ref{app:lbRingKnownNumber}. 
\begin{proposition} \label{prop:KnownNumberClosedRing}
	When locations are unknown, for any $\epsilon > 0$, there is no $(2-\epsilon)$-competitive algorithm for closed OLTSP on the ring, even if the number of requests is known.
\end{proposition}

Using the location data of the requests, we now describe an algorithm beating the competitive ratio of 2 in polynomial time.
For a simpler exposition, we assume that the shortest time to visit all request locations when ignoring release times is 1. This means that the shortest way to visit all requests is a round trip on the ring. Otherwise, the instance can be interpreted as an instance on the line. 
	
We illustrate the idea of the algorithm in Figure~\ref{fig:RingAlgo}. It first deals with a specific case, when there is a large interval without any request inside (so $[p_i,p_{i+1}]$ is large for some $i$). This is shown on the left in the figure. In this case the server goes to $p_i$ or $p_{i+1}$ (the farmost from $O$). Afterwards, it moves back to $O$, serving requests along the way and waiting for their release if necessary. Then, it goes to the other extremity. There, it turns around, moving back to $O$, again serving the requests on the way, waiting for their release if necessary. %\michalis{It then moves from $O$ to  the other request ($p_i$) and then back to $O$. Right?}
For the other case, shown on the right in the figure, when there is no such large empty interval, the server acts as follows: First, it waits at the origin to see whether to take a clockwise or a counter-clockwise tour around the ring. In order to decide this, the server waits for a contiguous part of at least $1/3$ of the ring to be released in either half of the ring. More precisely, all consecutive requests inside a segment of length at least $1/3$ are released and the whole segment lies completely in $(0,1/2]$ or in $[1/2, 1)$. We denote the time that this occurs by $t^{(1)}$. The server then starts moving on the shortest path to segment $s$, serving released requests it encounters on the way. It then continues in the same direction and serves segment $s$ (which has been completely released). Having visited $s$, the server continues in the same direction and stops at the first unreleased request on the tour around the ring. From then on, the server always waits for the request at the current location to be released before moving on to the next unserved request. Once the server reaches the origin, there might be some requests which have not been served because they were unreleased when the server passed their locations. To complete the TSP tour, the server will therefore go to the furthest unserved request and back to the origin, waiting at any unreleased request if necessary. 
\begin{figure}[ht]
	\begin{center}
		\begin{tikzpicture}[scale=0.7]
			\begin{scope}
			\draw[dotted] (0,0) arc (270:-90:2.5);  %the syntax is: Start (start:end (in degrees):radius), where a degree of 0 (mod 360) is the right point, so our origin at the bottom is at 270 degrees
			\draw[thick] ({2.5*cos(190)},{2.5 + 2.5*sin(190)}) arc [start angle=190, delta angle=210, radius=2.5cm]; %different way of writing the above, where delta angle is now not the end, but the "length"
			\node[] at (0,-0.5) (0) {$O$};
			\draw (0,0.25) -- (0,-0.25);
%			
%			\node[] at ({2.85*cos(210) }, {2.5+2.85*sin(210) }) (0) {$A$};
%			\draw ({2.5*cos(210)},{2.5+2.5*sin(210)}) -- ({2.7*cos(210)},{2.5+2.7*sin(210)});
			\node[] at ({2.9*cos(190) }, {2.5+2.9*sin(190) }) (0) {$p_i$};
			\draw ({2.2*cos(190)},{2.5+2.2*sin(190)}) -- ({2.7*cos(190)},{2.5+2.7*sin(190)});
			\node[] at ({3*cos(40) }, {2.5+3*sin(40) }) (0) {$p_{i+1}$};
			\draw ({2.2*cos(40)},{2.5+2.2*sin(40)}) -- ({2.8*cos(40)},{2.5+2.8*sin(40)});
			\draw[->, densely dashed] ({2.5*cos(88)-0.2*cos(88)},{2.5-2.5*sin(88)+0.2*sin(88)}) arc [start angle = 272, delta angle = 126, radius=2.3cm] node[midway, left] {\small{1}};
			\draw[<-, densely dashed] ({2.5*cos(88)-0.2*cos(88)},{2.5-2.5*sin(88)-0.2*sin(88)}) arc [start angle = 272, delta angle = 126, radius=2.7cm] node[midway, right] {\small{2}};
			\draw[<-, densely dashed] ({2.5*cos(192)-0.2*cos(192)},{2.5+2.5*sin(192)-0.2*sin(192)}) arc [start angle = 192, delta angle = 76, radius=2.3cm] node[midway, right] {\small{3}};
			\draw[->, densely dashed] ({2.5*cos(192)+0.2*cos(192)},{2.5+2.5*sin(192)+0.2*sin(192)}) arc [start angle = 192, delta angle = 76, radius=2.7cm] node[midway, left] {\small{4}};
%			\node[] at ({2.85*cos(30) }, {2.5+2.85*sin(30) }) (0) {$D$};
%			\draw ({2.5*cos(30)},{2.5+2.5*sin(30)}) -- ({2.7*cos(30)},{2.5+2.7*sin(30)});
%			\node[] at ({2.85*cos(-30) }, {2.5+2.85*sin(-30) }) (0) {$E$};
%			\draw ({2.5*cos(-30)},{2.5+2.5*sin(-30)}) -- ({2.7*cos(-30)},{2.5+2.7*sin(-30)});
		\end{scope}
		\begin{scope}[xshift=300]

			\draw[] (0,0) arc (270:-90:2.5);  %the syntax is: Start (start:end (in degrees):radius), where a degree of 0 (mod 360) is the right point, so our origin at the bottom is at 270 degrees
			\draw[thick] ({2.5*cos(320)},{2.5 + 2.5*sin(320)}) arc [start angle=320, delta angle=120, radius=2.5cm] node[midway,right] {$s$};
%			\draw[thick] ({2.5*cos(190)},{2.5 + 2.5*sin(190)}) arc [start angle=190, delta angle=210, radius=2.5cm]; %different way of writing the above, where delta angle is now not the end, but the "length"
			\node[] at (0,-0.5) (0) {$O$};
			\draw (0,0.25) -- (0,-0.25);
			\node[] at ({2.9*cos(80) }, {2.5+2.9*sin(80) }) (0) {};
			\draw ({2.2*cos(80)},{2.5+2.2*sin(80)}) -- ({2.8*cos(80)},{2.5+2.8*sin(80)});
			\node[] at ({3*cos(320) }, {2.5+3*sin(320) }) (0) {};
			\draw ({2.2*cos(320)},{2.5+2.2*sin(320)}) -- ({2.8*cos(320)},{2.5+2.8*sin(320)});
			
			\draw[->, densely dashed] ({2.5*cos(88)-0.2*cos(88)},{2.5-2.5*sin(88)+0.2*sin(88)}) arc [start angle = 272, delta angle = 46, radius=2.3cm] node[midway, above=-2] {\small{1}};
			\draw[->, densely dashed] ({2.5*cos(322)-0.2*cos(322)},{2.5+2.5*sin(322)-0.2*sin(322)}) arc [start angle = 322, delta angle = 116, radius=2.3cm] node[midway, left] {\small{2}};
			\draw[->, densely dashed] ({2.5*cos(82)+0.2*cos(82)},{2.5+2.5*sin(82)+0.2*sin(82)}) arc [start angle = 82, delta angle = 186, radius=2.7cm] node[midway, left] {\small{3}};
			\draw[->, densely dashed] ({2.5*cos(88)-0.2*cos(88)},{2.5-2.5*sin(88)+0.6*sin(88)}) arc [start angle = 272, delta angle = 46, radius=1.9cm] node[midway, above=0.2] {\small{4}};
			\draw[<-, densely dashed] ({2.5*cos(88)-0.2*cos(88)},{2.5-2.5*sin(88)-0.2*sin(88)}) arc [start angle = 272, delta angle = 46, radius=2.7cm] node[midway, below right] {\small{5}};
%			\draw[<-] (0,-0.2) arc [start angle = 270, delta angle = 130, radius=2.7cm] node[midway, right] {$2$};
%			\draw[<-] ({2.5*cos(190)-0.2*cos(190)},{2.5+2.5*sin(190)-0.2*sin(190)}) arc [start angle = 190, delta angle = 80, radius=2.3cm] node[midway, right] {$3$};
%			\draw[->] ({2.5*cos(190)+0.2*cos(190)},{2.5+2.5*sin(190)+0.2*sin(190)}) arc [start angle = 190, delta angle = 80, radius=2.7cm] node[midway, left] {$4$};
			%			\node[] at ({2.85*cos(30) }, {2.5+2.85*sin(30) }) (0) {$D$};
			%			\draw ({2.5*cos(30)},{2.5+2.5*sin(30)}) -- ({2.7*cos(30)},{2.5+2.7*sin(30)});
			%			\node[] at ({2.85*cos(-30) }, {2.5+2.85*sin(-30) }) (0) {$E$};
			%			\draw ({2.5*cos(-30)},{2.5+2.5*sin(-30)}) -- ({2.7*cos(-30)},{2.5+2.7*sin(-30)});
		\end{scope}
		\end{tikzpicture}
		\caption{The two main cases of Algorithm~\ref{algo:ring2}: On the left, there are two subsequent requests at positions $p_i, p_{i+1}$ with distance at least $1/3$ from each other, on the right, no such pair of requests exists.
		The movement of the algorithm is indicated by the arrows, where arrows on the inside of the ring show a movement without waiting and an arrow on the outside implies waiting at unreleased requests.	
		For the first case, the algorithm moves without waiting to position $p_{i+1}$ and then goes back to the origin, serving the requests and waiting for unreleased ones, if necessary. It then moves instantly to position $p_i$, serving the requests on the way back to the origin, again waiting if necessary.		
		For the second case, the algorithm waits until a segment $s$ of size at least $1/3$ (marked in bold) is fully released, then instantly goes to the closest endpoint of $s$, serving all released requests it encounters. Continuing in the same direction, it serves the requests in $s$ (which are all released). Afterwards, it continues in the same direction towards the origin, serving the requests it encounters, waiting if necessary. Then, it moves instantly to the furthest unserved request (which might be at the start of $s$), waits until it can be served and moves back to the origin, serving the unserved requests, waiting if necessary. }
		\label{fig:RingAlgo}
	\end{center}
\end{figure}
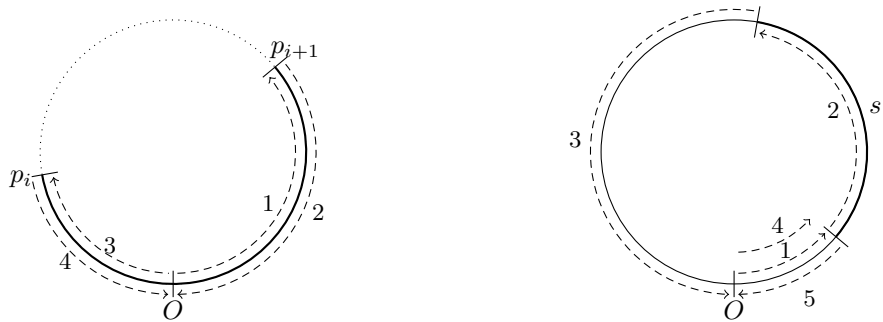%
For pseudocode, see Algorithm~\ref{algo:ring2}. The algorithm guarantees a competitive ratio of $5/3$, which we formalize in Theorem~\ref{theo:ringClosedCompetitiveRatio2}.

\begin{algorithm}[ht]
	\KwIn{Offline: Request locations $p_1, \dots, p_n$,  numbered clockwise from $O$\\
	\phantom{\textbf{Input:} }Online: Release times $t_1, \dots, t_n$}
	\eIf{there exists a segment  $s=[p_i,p_{i+1}]$ (with $1\leq i\leq n-1$) of size at least 1/3}{
	Go to the farmost request from $O$ between $p_i$ and $p_{i+1}$, say $p_{i+1}$.
	
	Move from $p_{i+1}$ to $O$ clockwise. Serve the requests on the way. If a request is unreleased, wait until it is released.
	
	Move from $O$ to $p_i$ (clockwise). Then go back to $O$, serving the requests on the way. If a request is unreleased, wait until it is released.}
	{
	Wait at the origin until $t^{(1)}$, when a segment $s \subseteq (0,1/2]$ or $s \subseteq [1/2,1)$ of length $|s| \ge 1/3$ is completely released, i.e., $\forall i \in [n]$ s.t. $p_i \in s$ it holds that $t_i \le t^{(1)}$. 
	If $s \subseteq (0,1/2]$, move clockwise, else, move counter clockwise, till the farmost (from $O$) extremity of $s$ is reached, serving requests along the way without waiting.\\
	Continue in the same direction back to the origin, serving requests along the way, waiting if necessary. Having reached $O$, continue in the same direction and stop at the unserved request which is furthest from the origin. Wait until it can be served, turn around and move back to the origin, on the way serving the remaining requests, waiting if necessary.
    }
    \caption{Algorithm for closed OLTSP-L on the ring}
	\label{algo:ring2}
\end{algorithm}

\begin{theorem}\label{theo:ringClosedCompetitiveRatio2}
	Algorithm~\ref{algo:ring2} has a competitive ratio of $5/3$ for closed OLTSP-L on the ring.
\end{theorem}

\begin{proof}
In both cases (whether there exists $[p_i,p_{i+1}]$ of size at least 1/3 or not), the algorithm goes to the farthest extremity of a free or completely released segment $s$ with extremities $P$ and $Q$, say point $P$, and then continues from this point to serve the (remaining) requests. Note that as the segment has size at least 1/3, $Q$ is at distance at most 1/3 from $O$. We first consider the case where $\alg$ has to wait at some time after reaching $P$. Let $t^{(0)}$ denote the last time that $\alg$ has to wait at an unreleased request $p_{j^*}$ and let $x$  be the distance between $O$ and $p_{j^*}$. Then, $|\opt| \ge t^{(0)} +x$. We distinguish two cases:
    \begin{itemize}
        \item If $x> 1/3$, and $Q$ and $p_{j^*}$ are on the same side of the ring (both in $[0,1/2]$ or both in $[1/2,1]$):  then $Q$ is at distance at most $1/6$ from $O$, and we have $|\alg| \le t^{(0)} + 1-x + 2 \cdot 1/6\leq t^{(0)}+1$, and $|\opt| \ge t^{(0)} +1/3$, so   $|\alg|\leq |\opt|+2/3\leq 5|\opt|/3$.
        \item Otherwise, as $Q$ is at distance at most $1/3$ from $O$, we have $|\alg| \le  t^{(0)}+x+2 \cdot 1/3$, and again  $|\alg|\leq |\opt|+2/3\leq 5|\opt|/3$.
    \end{itemize}

In the rest of the proof, we consider that $\alg$ does not wait after reaching the farmost extremity of the free/fully released segment.

    We first consider the first case, when there are two consecutive requests $p_i$, $p_{i+1}$ at distance at least 1/3. Let us denote by $x_j$ the distance from $p_j$ to $O$. Note that if $p_{i+1}$ is farther from $O$ than $p_i$ (the other case being symmetrical), we have $x_i\leq 1/3$.
    As $\alg$ never waits for an unreleased request, we have $|\alg|\leq 1+2x_i\leq 5/3 \leq 5|\opt|/3$. %\niklas{Shouldn't it be $|\alg| \le 1+2x_i$ or $|\alg| = \min \{1,2x_{i+1}\}$?}

Now we deal with the second case, i.e., there are no consecutive requests with a distance at least 1/3.

We first note that in this case any solution that does {\it not} go through all the points of the circle (i.e., does not make the whole tour) has cost at least $2 \cdot 2/3=4/3$.

%We first consider the case where $\alg$ has to wait for an unreleased request. As $|s|\geq 1/3$, the very same argument as in the first case, showing a ratio 5/3 when $\alg$ had to wait, holds. 

As $\alg$ does not wait for an unreleased request, we have 
\begin{equation}\label{eq:ring}
|\alg| \leq t^{(1)} + 1 + 2\cdot1/6 = t^{(1)} + 4/3\enspace.
\end{equation}
Indeed, after finishing a first tour the server may have to go to the closest extremity of $s$ and back to $O$, but as $|s|\geq 1/3$  the closest extremity of $s$ to $O$ is at distance at most $1/6$ from $O$.

Consider Figure~\ref{fig:circle}, and let us now look at $\opt$. As previously said, $|\opt|\geq 1$. Also, at $t^{(1)}-\epsilon$ there must be an unreleased request in $[A,C]$, and one in $[C,E]$. To serve both requests and go back to $O$, $\opt$ needs at least $1/2$ after $t^{(1)}$, so $|\opt|\geq t^{(1)}+1/2.$

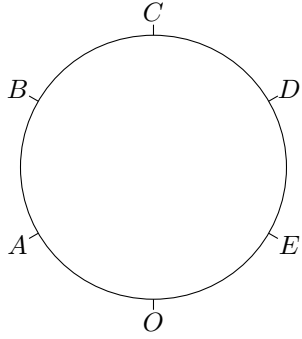
\begin{figure}[ht]
\begin{center}
\begin{tikzpicture}[scale=0.7]
	\draw (0,0) arc (270:-90:2.5);  %the syntax is: Start (start:end (in degrees):radius), where a degree of 0 (mod 360) is the right point, so our origin at the bottom is at 270 degrees
	%\draw[thick] (0,0) arc [start angle=270, delta angle=160, radius=2.5cm]; %different way of writing the above, where delta angle is now not the end, but the "length"
	\node[] at (0,-0.425) (0) {$O$};
    \draw (0,0) -- (0,-0.2);
	
	\node[] at ({2.95*cos(210) }, {2.5+2.95*sin(210) }) (0) {$A$};
 \draw ({2.5*cos(210)},{2.5+2.5*sin(210)}) -- ({2.7*cos(210)},{2.5+2.7*sin(210)});
\node[] at ({2.95*cos(150) }, {2.5+2.95*sin(150) }) (0) {$B$};
\draw ({2.5*cos(150)},{2.5+2.5*sin(150)}) -- ({2.7*cos(150)},{2.5+2.7*sin(150)});
\node[] at ({2.95*cos(90) }, {2.5+2.95*sin(90) }) (0) {$C$};
\draw ({2.5*cos(90)},{2.5+2.5*sin(90)}) -- ({2.7*cos(90)},{2.5+2.7*sin(90)});
\node[] at ({2.95*cos(30) }, {2.5+2.95*sin(30) }) (0) {$D$};
\draw ({2.5*cos(30)},{2.5+2.5*sin(30)}) -- ({2.7*cos(30)},{2.5+2.7*sin(30)});
\node[] at ({2.95*cos(-30) }, {2.5+2.95*sin(-30) }) (0) {$E$};
\draw ({2.5*cos(-30)},{2.5+2.5*sin(-30)}) -- ({2.7*cos(-30)},{2.5+2.7*sin(-30)});
\end{tikzpicture}
   \caption{Ring with points at clockwise distance $i/6$ from $O$, $i=1,\dots,5$}
    \label{fig:circle}
\end{center}
\end{figure}

We distinguish several cases in the following and show that in all cases either $|\opt|\geq t^{(1)}+2/3$, or $|\opt|\geq 4/3$. Note that then the ratio 5/3 follows: if $|\opt|\geq t^{(1)}+2/3$ then $|\alg|\leq t^{(1)}+4/3\leq |\opt|+2/3\leq 5|\opt|/3$, and if $|\opt|\geq 4/3$ then $|\alg|\leq t^{(1)}+4/3\leq |\opt|+(4/3-1/2)= |\opt|+5/6\leq 5|\opt|/3$ as $5/6\leq 2/3 \cdot 4/3\leq 2|\opt|/3$.

Let $F$ be the position of $\opt$ at time $t^{(1)}$, and $x_F$ its distance from $O$. We suppose, without loss of generality, that $F$ is in the clockwise first half of the ring (left part in Figure~\ref{fig:circle}). 

\begin{itemize}
    \item 
Suppose that $F$ is in $[OB]$. 
\begin{itemize}
    \item If $\opt$ goes through $C$ after $t^{(1)}$, then $|\opt|\geq t^{(1)}+1/6+1/2= t^{(1)}+2/3$.
    \item If $\opt$ goes through $C$ before $t^{(1)}$, and never again after $t^{(1)}$, then it is in $F$ at the earliest at $1-x_F$. Since there is an unreleased request (at $t^{(1)}-\epsilon$) in $[C,E]$, from $F$ to reach this request (without going through $C$) and to go back to $O$, $\opt$ needs at least $x_F+2 \cdot 1/6$, so $|\opt|\geq 1+1/3=4/3$.
    \item If $\opt$ never goes through $C$ then $|\opt|\geq 4/3$.
\end{itemize}
\item Suppose now that $F$  is in $[BC]$.
\begin{itemize}
\item Suppose that $\opt$ goes through $D$ after $t^{(1)}$.  
\begin{itemize}
    \item If $\opt$ goes back into $[OB]$ before $D$, then $|\opt|\geq t^{(1)}+2/3$ (1/3 from $B$ to $D$, 1/3 from $D$ to $O$).
    \item If $\opt$ goes to some point on $[AC]$ after going through $D$, then $|\opt|\geq t^{(1)}+2/3$ (1/2 from $C$ to $O$ and 1/3>1/6 from $O$ to $A$ and back to $O$ -- or going from $F$ to $D$, which takes at least $1/6$ and back to $O$ through $[AC]$, requiring again at least $2/3$).
    \item Otherwise: there cannot be an unreleased request in $[AB]$ (as $\opt$ has to go there, and we are in one of the two above cases). Let $x$ be the position of the farthest request in $[OB]$ unreleased at $t^{(1)}-\epsilon$. This request is in $[OA]$. Then there is an unreleased request at position $y$ in $[x,x+1/3]$. This second request must be in $[BC]$. Then, $\opt$ has to go to position $y$, then to $D$, then to $x$, and back to $O$. We get $|\opt|\geq t^{(1)}+(1-y)+2x$. As $x\geq y-1/3$ and $y\geq 1/3$, we have $|\opt|\geq t^{(1)}+1/3+y\geq t^{(1)}+2/3$.  
\end{itemize}  
\item Suppose finally that $\opt$ never goes through $D$ after $t^{(1)}$. If there is some unreleased request in $[DE]$,  then it needs at least 1/3 to reach $O$, and $2 \cdot 1/6$ to serve the unreleased request on $[DE]$, so $|\opt|\geq t^*+2/3$. Otherwise, there are two unreleased request, one in $[EO]$, one in $[CD]$, at distance at most 1/3. $\opt$ must serve the one in $[CD]$, then take the long road to serve the one in $[EO]$, and this long road has length at least 2/3. So $|\opt|\geq t^{(1)}+2/3$.
\end{itemize}
\end{itemize}
%\niklas{Looks good, some minor comments: There are many different notations for segments: $(A,B), (AB), AB$. We should probably stick to one. (And maybe define it first?) Also, in the last paragraph, I was a bit confused at first because it's "If there is some unreleased request in (D,E), ... to serve the unreleased request on (CE)" Technically, that's correct, since $(DE) \subseteq (CE)$ but I was trying to see why it's now the bigger segment which confused me at first.}

%Let $N$ denotes the point at distance $1/3$ from $O$ which is not on the same side of $A$.\\

%{\bf Case 1.} $\opt$ never goes through $N$ after $t^{(1)}$. As previously shown, $|\opt|\geq t^{(1)}+1/2$. Now, if $\opt$ never goes through $N$ at all, as mentioned above we have $|\opt|\geq 4/3$. If $\opt$ went through $N$ before $t^{(1)}$, as it never goes through $N$ after that time but still have one request at distance at least 1/6 from O on each side, $\opt$ makes at best a whole turn plus $2. 1/6$, and again $|\opt|\geq 4/3$.

%Finally, $|\alg|\leq t^{(1)}+4/3= t^{(1)}+1/2+5/6\leq |\opt|+15|\opt|/24<5|\opt|/3$.\\

\end{proof}

\section{Stars} \label{Stars}

We consider in this section the case of stars. There is a star, centered at $O$ (the initial position), with $k$ branches or rays. If two requests $q_i$ and $q_j$, located at $p_i$ and $p_j$, are not in the same branch, then the distance between them is $d(p_i,p_j)=d(p_i,O)+d(O,p_j)$. 
We first show a lower bound of $2$ when the locations of the requests are unknown (even if their number is known), in both the closed and the open version. Formally, we prove the following (the proof can be found in Appendix~\ref{app:KnownNumberStar}).

\begin{proposition} \label{prop:KnownNumberStar}
When the locations are unknown, for any $\epsilon>0$, there is no $(2-\epsilon)$-competitive algorithm for both closed and open OLTSP on the star, even if the number of requests is known.
\end{proposition}

We then focus on the closed case and show that this bound of 2 can be improved on stars when the locations of requests are known, with a polynomial time algorithm. More precisely, we devise a polytime algorithm which is $(7/4+\epsilon)$-competitive (in the closed case on stars). 

The idea of the algorithm is the following. If a long ray (i.e., ray of length at least $1/4$ of the overall length of all rays)  exists, it already gives a good competitive ratio to first serve this ray completely before going over to the remaining rays. Otherwise, if all rays are short, the algorithm waits in the origin until time $t$, which is exactly the combined length of all rays. At this time, it identifies a set $R$ of rays maximizing the released segments in $R$ under the constraint that the set $R$ can be traversed completely, including going back to the origin, in time $t$. For this, only contiguous segments starting at the outer extremities of the rays are counted. If a segment does not start at the outer extremity of a ray, it will have to be traversed to reach the extremity, so serving it early does not help. Having identified the set $R$, our algorithm then serves the released requests in the set $R$ and waits at the origin until all requests are released. Afterwards, it serves the unserved requests in an optimal manner. A visualization of the algorithm is given in Figure~\ref{fig:starExample}. Note that finding an optimal set $R$ constitutes solving a knapsack problem. Thus, this might not be possible in polynomial time. For a simpler exposition, we first show that with an optimal $R$, the algorithm achieves a competitive ratio of $7/4$ in Theorem~\ref{theo:starClosedCompetitiveRatio}. Utilizing a knapsack FPTAS to find $R$ gives us, for every constant $\epsilon > 0$, a $7/4+\epsilon$-competitive polynomial time algorithm (see Corollary~\ref{cor:StarClosedCompetitiveRatioPolyTime}).

\begin{algorithm}[ht]
	\KwIn{Offline: Request locations $p_1, \dots, p_n$\\
	\phantom{\textbf{Input:} }Online: Release times $t_1, \dots, t_n$}
	Let $R_1,\dots,R_k$ be the rays of the star and $r_1, \dots, r_k$ the lengths of the rays, i.e., the distance from the origin to the respective extremity of the ray. \\
	\uIf{there exists $r_j, j \in [k]$, with $r_j \ge 1/4 \sum_{j'=1}^{k} r_{j'}$}{Instantly traverse ray $R_j$, turning once reaching the extremity and start to serve the requests. If a request is unreleased, wait at this location until it is released before continuing.}
	\Else{ Wait in the origin until $t = \sum_{j'=1}^k r_{j'}$.\\
		Let $R$ be the set of indices $R \subseteq [k]$ maximizing the combined value of contiguously released segments, starting from the outer extremities of their respective rays, s.t.\  $\sum_{j' \in R} r_{j'} \le 1/2 \sum_{j' = 1}^{k} r_{j'}$.\\
		Traverse the rays of set $R$, serving the released requests and return to $O$.}
	Afterwards, wait in the origin until all requests are released and serve the unserved requests in an optimal manner.
	\caption{Algorithm for closed OLTSP-L on the star}\label{algo:starClosedWaitFirst}
\end{algorithm}

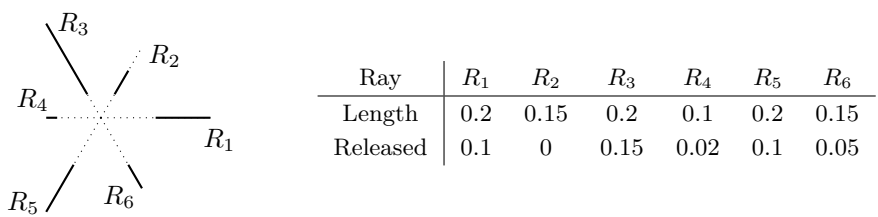
\begin{figure}[h]
	\centering
\begin{minipage}{0.3\textwidth}

\begin{tikzpicture}[scale=1.2]
	\draw[dotted] (0,0) -- (1.2,0);
	\draw[dotted] (0,0) -- ({0.9*cos(60)},{0.9*sin(60)});
	\draw[dotted] (0,0) -- ({1.2*cos(120)},{1.2*sin(120)});
	\draw[dotted] (0,0) -- ({0.6*cos(180)},{0.6*sin(180)});
	\draw[dotted] (0,0) -- ({1.2*cos(240)},{1.2*sin(240)});
	\draw[dotted] (0,0) -- ({0.9*cos(300)},{0.9*sin(300)});
	\node at (1.3,-0.25) (R1) {$R_1$};
	\node at ({0.9*cos(60)+0.25},{0.9*sin(60)-0.1}) (R2) {$R_2$};
	\node at ({1.2*cos(120)+0.3},{1.2*sin(120)}) (R3) {$R_3$};
	\node at ({0.6*cos(180)-0.15},{0.6*sin(180)+0.2}) (R4) {$R_4$};
	\node at ({1.2*cos(240)-0.25},{1.2*sin(240)+0.1}) (R5) {$R_5$};
	\node at ({0.9*cos(300)-0.25},{0.9*sin(300)-0.1}) (R6) {$R_6$};
	
	\draw[thick] (0.6,0) -- (1.2,0);
	\draw[thick] ({0.3*cos(60)},{0.3*sin(60)}) -- ({0.6*cos(60)},{0.6*sin(60)});
	\draw[thick] ({0.3*cos(120)},{0.3*sin(120)}) -- ({1.2*cos(120)},{1.2*sin(120)});
	\draw[thick] ({0.48*cos(180)},{0.48*sin(180)}) -- ({0.6*cos(180)},{0.6*sin(180)});
	\draw[thick] ({0.6*cos(240)},{0.6*sin(240)}) -- ({1.2*cos(240)},{1.2*sin(240)});
	\draw[thick] ({0.6*cos(300)},{0.6*sin(300)}) -- ({0.9*cos(300)},{0.9*sin(300)});
	
\end{tikzpicture}

\end{minipage}%
\begin{minipage}{.6\textwidth}
\begin{tabular}{c|cccccc}
	Ray & $R_1$ & $R_2$ & $R_3$ & $R_4$ & $R_5$ & $R_6$ \\
	\hline 
	Length & 0.2 & 0.15 & 0.2 & 0.1 & 0.2 & 0.15 \\
	Released & 0.1 & 0 & 0.15  & 0.02  & 0.1  & 0.05
\end{tabular}
\end{minipage}
\caption{There are 6 rays in the star. The length of each ray indicates the farthest request on it. Released segments are visualized using bold lines. The lengths of the rays and the released segments are given in the table. Although there is a segment which is released on $R_2$, the released length is 0, as this segment does not start at the outer extremity of $R_2$. For this situation, $R$ would consist of the rays $R_1, R_3$ and $R_4$ (or $R_3, R_4, R_5$) with a combined length of $1/2$ and $\ell = 0.27$, where $\ell$ is the combined length of released segments in $R$.}\label{fig:starExample}
\end{figure}

\begin{theorem}\label{theo:starClosedCompetitiveRatio}
	Algorithm~\ref{algo:starClosedWaitFirst} achieves a competitive ratio of $7/4$ for closed OLTSP-L on the star.
\end{theorem}
\begin{proof}
	We assume without loss of generality that $\sum_{j=1}^k r_j = 1$. This means that $|\opt| \ge 2$, since every ray has to be fully traversed and the server has to come back to the origin. Further, for the analysis we assume w.l.o.g.\ that requests are served on the way back to the origin, and that a request is served only if all requests which are farther from the origin on the same ray have already been served.
	
	For the first case, assume that there is a ray $r$ of length at least $1/4$. Then, the algorithm instantly goes to the far end of the ray and serves its requests on the way back, waiting at each until it is released. Having returned to the origin at time $t^*$, the algorithm can serve the remaining requests in time $3/2$ once everything is released. It further holds that $|\opt| \ge t^*$, i.e., $\opt$ cannot serve all requests and return to the origin before $\alg$ has served the ray $r$ and returned to the origin. Since $\alg$ starts serving the remaining requests when all requests are released and $\opt$ cannot finish before that time, we get $|\alg| \le |\opt| + 3/2$. This means $\ratio \le 1+3/4 = 7/4$.
	
	Otherwise, if all rays have a length of less than $1/4$, the algorithm waits until time $t=1$. At this time, it identifies the set $R$ and serves the requests in $R$. Since the combined length of rays in $R$ is at most $1/2$, $\alg$ will be back at the origin by time $t=2$. Denote the combined length of the released segments in the rays with index in $R$ by $\ell$. $\alg$ then serves the requests and returns to the origin, waiting there until all remaining requests are released, which is at most $t = |\opt|$. For the remaining requests, $\alg$ needs time at most $2-2\ell$, so we get $|\alg| \le |\opt| + 2-2\ell$. 
	Clearly, if $\ell \ge 1/4$, this means that $|\alg| \le |\opt| + 3/2 \le 7/4\cdot |\opt|$ as $|\opt| \ge 2$.
	
	Otherwise, if $\ell < 1/4$, we first observe that even though at most $\ell$ can be covered by time $t=1$, this does not mean that $\opt$ can have served only an $\ell$-fraction of the requests. It could be that $\opt$ is already on some ray $R_j$, having visited some requests, but not able to return to the origin by time $t=1$. However, by the definition of $R$, $\opt$ can have visited at most $\ell$ before traversing $R_j$ (the ray on which it is at $t=1$). Denote by $\ell'$ the length of the segment of requests visited by $\opt$ on $R_j$ at time 1. It must hold that $\ell' \le \ell$. Thus, $\opt$ must serve $r_j - \ell'$ on $R_j$ and at least $1-r_j - \ell$ on the remaining rays. Since $r_j \le 1/4$, this implies 
	\[
	|\opt| \ge 1 + r_j - \ell' + 2(1-r_j - \ell) \ge 3- r_j - 3\ell \ge 11/4 - 3\ell \enspace. \] 
	Since $|\alg| \le |\opt| + 2 - 2 \ell$, we have $\ratio \le 1 + \frac{2-2\ell}{11/4-3\ell} \le 1+ \frac{3}{4}$ as $ \frac{2-2\ell}{11/4-3\ell} \le \frac{3}{4}$ for $0 \le \ell \le 1/4$.
	This concludes the proof.

\end{proof}

Note that for a constant number $k$ of rays in the star, the algorithm is able to find the set $R$ in polynomial time and is thus able to run in polynomial time. Otherwise, we can use a knapsack FPTAS to approximately solve for the best set $R$ in polynomial time.

\begin{corollary}\label{cor:StarClosedCompetitiveRatioPolyTime}
    Using a knapsack-FPTAS to find set $R$ in Algorithm~\ref{algo:starClosedWaitFirst}, the algorithm achieves a competitive ratio of $7/4 + \epsilon$ for any constant $\epsilon > 0$ for closed OLTSP-L on the star in polynomial time.
\end{corollary}
\begin{proof}
    The set $R$ returned by the knapsack FPTAS (see, e.g., \cite[Section 6.2]{KellererPP04}) guarantees a total released length of $(1-\epsilon)\ell$ for any constant $\epsilon > 0$, where $\ell$ is the optimal released length. Thus, we get $|\alg| \le |\opt| + 2-2\ell + 2\ell\epsilon$, which increases $\ratio$ and thus the competitive ratio by at most an additive $\epsilon$ as $\opt \ge 2$ and $\ell \le 1$.
\end{proof}

\section{Semi-Line} \label{SemiLine}
In this section, we study the case of a semi-line.
Without loss of generality, we consider a semi-line starting at the origin $O$. We denote the location of requests by their distance from the origin, the distance of the furthest request by $L$ and say that the semi-line has length $L$.

First, we show the following lower bounds when only the number of requests is known to an algorithm for both the closed and the open variant (the formal proofs are in Appendix~\ref{app:KnownNumberClosedSemiline} and~\ref{app:KnownNumberOpenSemiline}, respectively).

\begin{proposition} \label{prop:KnownNumberClosedSemiline}
When locations are unknown, for any $\epsilon > 0$, there is no $(4/3-\epsilon)$-competitive algorithm for closed OLTSP on the semi-line, even if the number of requests is known.
\end{proposition}

\begin{proposition} \label{prop:KnownNumberOpenSemiline}
When locations are unknown, for any $\epsilon > 0$, there is no $(3/2-\epsilon)$-competitive algorithm for open OLTSP on the semi-line, even if the number of requests is known.
\end{proposition}

Next, we discuss OLTSP-L on the semi-line. For the open version of the problem, we give a lower bound of $4/3$ on the semi-line (Theorem~\ref{theo:semilineOpenLB}) and then describe an algorithm that is $13/9$-competitive (Theorem~\ref{theo:semilineOpenCompetitiveRatio}). 

Let us first show the lower bound of $4/3$ on the semi-line.

\begin{theorem}\label{theo:semilineOpenLB}
	For any $\epsilon >0$, there is no $(4/3- \epsilon)$-competitive algorithm for open OLTSP-L on the semi-line.
\end{theorem}

\begin{proof}
    Let $\epsilon > 0$ and consider a semi-line with length $1$ starting at $O$. We place $4$ requests at distance $0$, $1/6$, $5/6$ and $1$ from $O$, respectively. Consider an algorithm $\alg$ serving the requests and let $s$ be the distance of $\alg$ from $O$.
    
    At time $t=1$:
    \begin{itemize}
        \item if $s\in[0,1/6)$ or $s\in(5/6,1]$, we release the furthest request $q$ from $\alg$ ($q$ is at $1$ or $0$, respectively). Then, we release the next requests, one after the other, starting from $q$ and ending on the opposite side at times $t=7/6$, $t=11/6$ and $t=2$, respectively.
        
        \item else: we release the two outer requests (at positions $0$ and $1$). We then release the furthest inner request (at $1/6$ or $5/6$) at time $t=7/6$. At time $t=11/6$, we release the final inner request.
    \end{itemize}
    
    Note that $\opt$ needs exactly  $2$ units of time to serve all requests. So, $|\opt|=2$. Without loss of generality, we can assume that $\alg$ is at $s\in [1/2,1]$ at time $t=1$. We now distinguish the following cases:
    \begin{itemize}
        \item $s \in (5/6,1]$ at $t=1$: the furthest outer request at $0$ is released at $t=1$ and the request at $1$ is released at time $t=2$. Since $\alg$ has to serve both these requests, it needs at least $2+5/6=17/6$ units of time and $|\alg| \geq 17/6$.
        
        \item $s \in [1/2, 5/6]$ at $t=1$: the two outer requests are released (at $0$ and $1$). We now study the following three cases separately according to the position of $\alg$ at time $t=7/6$. Note that $\alg$ cannot be at $s<1/3$ when $t=7/6$.
        
        \begin{itemize}
            \item $s \in [1/3,1/2]$ at $t=7/6$: the request at $5/6$ is released and $\alg$ has not served any requests yet. If $\alg$ moves to the right after $t=7/6$, then the earliest it can serve all the requests is $7/6+1/2+1=8/3$. If it moves to the left, then the earliest it can finish is $11/6 + 5/6=8/3$ ($11/6$ is the release time of the request at $1/6$). So, $|\alg| \geq 8/3$. 
            
            \item $s \in (1/2,5/6)$ at $t=7/6$: the request at $1/6$ is released and $\alg$ has not served any requests yet. If $\alg$ moves to the left after $t=7/6$, the the earliest it can finish is $7/6+1/2+1=8/3$. If it moves to the right, then the earliest it can serve all the requests is $11/6+5/6=8/3$ ($11/6$ is the release time of the request at $5/6$). So, $|\alg| \geq 8/3$. 
            
            \item $s \in [5/6,1]$ at $t=7/6$: the request at $1/6$ is released. If $\alg$ serves the request at 5/6 before the one at 0, then it needs at least 11/6+5/6=8/3. Otherwise, since $s\geq 5/6$ at $t=7/6$, $\alg$ cannot serve the request at 0 before 2, and then needs at least $2+5/6>8/3$.

            % If $\alg$ has already served the request at $1$ at time $t=7/6$, then $s\geq 5/6$. If $\alg$ serves next the request at position $5/6$, the best strategy is to wait until $t=11/6$ for its release. Then, $\alg$ cannot finish before $11/6+5/6=8/3$. If $\alg$ serves next the request at position $0$, it needs at least $7/6+5/6+1=3$ units of time to finish. The case when $\alg$ serves next the request at $1/6$ is trivial. So, $|\alg| \geq 8/3$.
            
            % If $\alg$ hasn't served any request yet at time $t=7/6$, then we have the following cases. If it serves first the request at position $1$, then the best strategy is to wait until $t=11/6$ at position $5/6$. In that case, $|\alg| \geq 11/6+5/6=8/3$. \michalis{actually this is the same case as before.} If it serves first the request at $0$, then $|\alg| \geq 7/6+5/6+1=3$. The cases when $\alg$ serves first one of the inner requests are trivial. So, again $|\alg| \geq 8/3$.
        \end{itemize}
    \end{itemize}
    Having $|\alg| \geq 8/3$ and $|\opt| = 2$ we get a lower bound of $4/3$.
\end{proof}

We are now ready to present an algorithm that is $13/9$-competitive for open OLTSP-L on the semi-line. The principle of this algorithm (see Algorithm~\ref{algo:semi-line} for a formal description) is to reach position $L/2$ and wait until segment $[0,L/4]$ or $[3L/4,L]$ is completely released. More precisely, it waits until all consecutive requests inside the segment $[0,L/4]$ or $[3L/4,L]$ are released. After that, the algorithm moves to position $0$ or $L$ and serves all requests on the semi-line without turning around again and possibly waiting multiple times. To keep the competitive ratio small, we have to be careful with how we move away from the origin to reach $L/2$. 

\begin{itemize}
    \item At first, $\alg$ moves from the origin only if it does not leave any unreleased requests behind. It follows the same strategy until $t=\min\{L/2+x,3L/4\}$, where $x$ is the furthest request position such that $[0,x)$\footnote{Strictly speaking, if $x \ge L$, all requests and thus $[0,x]$ would be released. We will ignore this and only use half-open intervals to avoid technicalities.} is completely released by time $t$. By construction, the contiguous segment starting from 0 which is served by $\alg$ is at least the size of that of $\opt$.
    
    \item If $x \geq L/4$, then $\alg$ continues the same strategy until it has served all requests. Otherwise, if $x<L/4$, $\alg$ is at position $x$ at $t=L/2+x$ and abandons that strategy to reach position $L/2$ at exactly $t=L$, where it waits for $[0,L/4]$ or $[3L/4,L]$ to be completely released as described above.
\end{itemize}

\begin{algorithm}[ht]
\DontPrintSemicolon
	\KwIn{Offline: Request locations $p_1, \dots, p_n$, where $p_1 \leq p_2 \leq ... \leq p_n$\\
	\phantom{\textbf{Input:} }Online: Release times $t_1, \dots, t_n$}
	$L \leftarrow d(O,p_n)$\\
	Starting from the origin, move in the direction of $p_n$ and serve all released requests that are encountered. Stop at the first unreleased request. Wait until it can be served, serve it and continue this procedure until $t=\min\{L/2+x,3L/4\}$, where $x$ is the furthest position such that $[0,x)$ is completely released by time $t$.\\
	\lIf{$x \ge L/4$}{
    continue in the same direction until all requests are served. The algorithm terminates.}
    \Else{Start moving to position $p=L/2$. \\
        \While{$t < L$ \emph{\textbf{and}} $[0,L/4]$ is completely released}{return to position $x$ and serve the requests in $[x,L/4]$.
        
        \textbf{break}
        }
	    Wait at $L/2$ until $[0,L/4]$ \big(or $[3L/4, L]$\big) is completely released and serve it moving from $x$ to $L/4$ \big(or from $L$ to $3L/4$\big). Continue in the same direction and stop at the first unserved request.
        Wait until it can be served and continue this procedure until all requests are served.
    }
    \caption{Algorithm for open OLTSP-L on the semi-line}
	\label{algo:semi-line}
\end{algorithm}

\begin{theorem}\label{theo:semilineOpenCompetitiveRatio}
	Algorithm~\ref{algo:semi-line} has a competitive ratio of $13/9$ for open OLTSP-L on the semi-line.
\end{theorem}

\begin{proof}

    Without loss of generality, let $t^{(1)}$ be the first time when $[0,L/4]$ is completely released (when $t^{(1)}<L$) or the first time when $[0,L/4]$ or $[3L/4,L]$ is completely released (when $t^{(1)}\geq L$).  
    
    The proof is based on the following observations.
    \vspace{0.2cm}

    \noindent {\it Observation 1. If $\opt$ turns around at least once, then, without loss of generality, the first time it turns is at $L$.}
    \vspace{0.2cm}

    Assume that this is not the case. This means that $\opt$ turns around at least twice, since it turned around before reaching $L$, where there is a request. Thus, $\opt$ traversed some segment $[s_1,s_2]$ twice. Instead of turning around at $s_2$ and one more time at $s_1$, we adapt the algorithm, calling it $\opt'$, by letting it wait at $s_1$ until $t^{\opt}_{s_1}$, the time at which $\opt$ serves $s_1$. Afterwards, we copy the behavior of $\opt$. This means that the completion time of our new algorithm $\opt'$ is the same as the completion time of $\opt$. By repeatedly applying this adaptation, we can make sure that $\opt$ never turns around before reaching $L$ without increasing the overall time.
    \vspace{0.2cm}

    \noindent {\it Observation 2. Assume that $\alg$ has already served $[0,x_{\alg})$ at $t=t^{(1)}$. If $t^{(1)} \leq L/2+x$ and $x<L/4$, then $\opt$ has served at most $[0,x_{\opt})$ at $t=t^{(1)}$, where $x_{\opt} \leq x_{\alg}$.} 
    \vspace{0.2cm}

    The above observation is true by construction. $\alg$ has served the maximal continuous segment $[0,a)$ that any algorithm can serve, for some $a\geq 0$, until $t^{(1)} \leq L/2+x$ (when $x<L/4$). 
    
    \vspace{0.2cm}

    \noindent {\it Observation 3. When $x<L/4$ and $t=L/2+x$, $\alg$ is at position $x$.} 
    \vspace{0.2cm}
    
    Assume that $\alg$ is at position $0 \le x_0 < L/4$ at time $t_0=L/2$. If $x_0=0$, then $x=x_0=0$ and this phase of the algorithm (line 2 in the description) terminates with $\alg$ at position $x=0$. Otherwise, $x_0>0$ and this phase runs until at least $t_1=L/2+x_0$. If at time  $t_1=L/2+x_0$ there is no completely released segment $[0,x_1)$ with $L/4>x_1>x_0$, then $x=x_0$, $t=L/2+x_0$ (since $x<L/4$) and the phase terminates with $\alg$ at position $x$. If at time $t_1=L/2+x_0$ there is such a fully released segment \big($[0,x_1)$ with $L/4 >x_1>x_0$\big), then this phase continues until at least $t_2=L/2+x_1$. By definition of the algorithm it holds that $\alg$ will have moved to $x_1$ by $t_2$ at the latest. The starting time is $t_1$ and the starting position is at least $x_0$. Thus, the distance is at most $x_1 - x_0$. The available time is $t_2 - t_1 = x_1 - x_0$, and whenever $\alg$ has not reached the end of a fully released segment, it will move with unit speed towards this point. 
    Now, either $x=x_1$ or there exists a larger completely released segment $[0, x_2)$ with $L/4 > x_2 > x_1$ and we can argue as above\dots
    
    Note that whenever a bigger segment exists, this is because at least one new request was released -- but there are only $n$ requests in total (and at least $q_n$ is outside of $[0, L/4)$). Hence, after $i <n$ many steps, no larger completely released segment $[0,x_{i+1})$ can exist such that $x_{i+1} < L/4$ and $\alg$ is at position $x = x_i$ at $t=L/2+x$.
    \vspace{0.2cm}

    We now distinguish the following cases:
    
    \begin{itemize}
        \item If $x \ge L/4$, then, without loss of generality, we assume that $\opt$ turns around at least one time. Otherwise, we have that $|\alg|=|\opt|$. So, from Observation $1$ we can assume that the first time that $\opt$ turns around is at $L$.
        
        Let $x_{\alg}$ and $x_{\opt}$ be such that $[0,x_{\alg})$ and $[0,x_{\opt})$ have been served by $\alg$ and $\opt$, respectively, when $\opt$ reaches $L$ the first time. We have that $x_{\alg} \geq x_{\opt}$ and $x_{\alg} \ge L/4$.
        
        Since $\opt$ turns around at $L$ we get $|\opt| \geq L + L - x_{\opt}=2L - x_{\opt}$. We also have that $|\alg| \leq |\opt| + L - x_{\alg}$ and, since $x_{\alg} \geq x_{\opt}$, it follows that
        % \begin{equation}
        % \label{eq:opt_1}
        %     |\opt| \geq L + L - x_{\opt}=2L - x_{\opt}
        % \end{equation}
        % \begin{center}
        %      and
        % \end{center}
    
        % \begin{equation}
        % \label{eq:alg_1}
        %     |\alg| \leq |\opt| + L - x_{\alg}\enspace.
        % \end{equation}
        
        % From (\ref{eq:opt_1}), (\ref{eq:alg_1}) and since $x_{\alg} \geq x_{\opt}$ we have that
        \begin{equation*}
            \ratio \leq \frac{|\opt| + L - x_{\alg}}{|\opt|}
            \leq 1 + \frac{L - x_{\alg}}{2L - x_{\opt}} \leq
            1+\frac{\frac{L}{x_{\alg}}-1}{\frac{2L}{x_{\alg}}- \frac{x_{\opt}}{x_{\alg}}}
            \leq 1+\frac{\frac{L}{x_{\alg}}-1}{\frac{2L}{x_{\alg}}- 1} \enspace.
        \end{equation*}
        
        Setting $a=\frac{L}{x_{\alg}}$, we have that $f(a) = \frac{a-1}{2a-1}$ is increasing at $a = \frac{L}{x_{\alg}} \in (1,4]$. Thus, for $a=4$ we get $\ratio \le 1 + \frac{4-1}{8-1} = 10/7.$
        % \begin{equation*}
        %     \ratio \leq 1 + \frac{4-1}{8-1} \leq 10/7 \enspace.
        % \end{equation*}
        
        \item If $x < L/4$ and $t^{(1)}\in [L/2+x,L)$, then $\alg$ has already served $[0,x)$ and is at position $(t^{(1)}-L/2)$ at $t=t^{(1)}$. It starts returning to position $x$. $\opt$ has already served $[0, x_{\opt})$ and is at position $s_{\opt} \leq t^{(1)}$ at time $t=t^{(1)}$. 
        
        \begin{itemize}
            \item If $s_{\opt} > t^{(1)}-L/2$, then $x_{\opt} \leq x$ and $\opt$ is closer to $L$ than $\alg$. Thus, $\opt$ has to turn around at least one time to serve $[x,x^+)$ (where $x^+$ is after $x$ on the semi-line but unknown) and we can assume that the first time it turns around is at $L$. Therefore,
            $|\opt| \geq t^{(1)} + L - s_{\opt} + L-x
                \geq 2L-x \text{, since $t^{(1)} \geq s_{\opt}$}.$
            % \begin{equation}
            % \label{eq:opt_2a}
            %     |\opt| \geq t^{(1)} + L - s_{\opt} + L-x
            %     \geq 2L-x \text{, since $t^{(1)} \geq s_{\opt}$} \enspace.
            % \end{equation}
            Since $x< L/4$, we have that $|\opt| > 7L/4.$
            % \begin{equation}
            % \label{eq:opt_2b}
            %     |\opt| > 7L/4 \enspace.
            % \end{equation}
            
            $\alg$ at $t=t^{(1)}$ starts returning to $x$ and has to serve $[x,L]$ and possibly wait for unreleased requests. So, we get
            \begin{align*}
                |\alg| &\leq \max\Bigl\{t^{(1)}+t^{(1)}-L/2-x+L/4-x, |\opt|\Bigl\}+3L/4 \\
                &\leq \max\Bigl\{7L/4-2x,|\opt|\Bigl\}+3L/4 = |\opt| + 3L/4 \enspace.
            \end{align*}
            
            Overall, this means that
            $\ratio  \leq 3/7 + 1 = 10/7 .$
            
            % \begin{equation*}
            %     \ratio  \leq 3/7 + \frac{\max\Bigl\{7L/4-2x,|\opt|\Bigl\}}{|\opt|} \leq 10/7 \enspace.
            % \end{equation*}
            
            \item Else: $s_{\opt} \leq t^{(1)}-L/2$ and $x_{\opt} < L/4$.
            
            \begin{itemize}
                \item If $\opt$ does not turn around at any time, then $s_{\opt} \leq x < L/4$ at time $t=t^{(1)}$, since $[0,d]$ with $d\ge L/4$ is not completely released before $t^{(1)}$, and $s_{\opt} = x_{\opt}$.
                
                We have that $|\opt| \geq t^{(1)} + L - s_{\opt} \geq t^{(1)} + 3L/4.$
                % \begin{equation}
                % \label{eq:opt_3a}
                %     |\opt| \geq t^{(1)} + L - s_{\opt} \geq t^{(1)} + 3L/4 \enspace.
                % \end{equation}
                
                In addition, we get
                \begin{equation*}
                    |\alg| \leq \max\Bigl\{t^{(1)}+t^{(1)}-L/2-x+L-x, |\opt|\Bigl\} 
                %\end{equation*}
                %\begin{equation}
                \label{eq:alg_3}
                    = 
                    \max\Bigl\{2t^{(1)}+L/2-2x, |\opt|\Bigl\} \enspace.
                \end{equation*}
                It follows with the assumption $t^{(1)} < L$  that
                \begin{align*}
                    \ratio
                    \leq \frac{\max\Bigl\{2t^{(1)}+L/2-2x, |\opt|\Bigl\}}{|\opt|}
                    %= \max\biggl\{ \frac{2t^{(1)}+L/2-2x}{|\opt|}, 1 \biggl\} \\
                    \leq \max\biggl\{ 1 + \frac{t^{(1)}-L/4-2x}{t^{(1)}+3L/4},
                    1 \biggl\}
                    \leq 10/7 \enspace.
                \end{align*}
                
                \item Else: $\opt$ turns around at least one time. Then, from Observation $1$, we can assume
                that the first time that $\opt$ turns around is at $L$. 
                
                So, from the assumption $s_{\opt} \leq t^{(1)}-L/2$ we get $|\opt| \geq t^{(1)} + L - s_{\opt} + 3L/4 \geq 9L/4$
                    and
                    $|\alg| \leq |\opt| + 3L/4$.
                Therefore, $\ratio
                    \leq 1+\frac{3L/4}{9L/4} = 4/3 .$
                % \begin{equation*}
                %     \ratio
                %     \leq 1+\frac{3L/4}{9L/4} = 4/3 \enspace.
                % \end{equation*}
            \end{itemize}
        \end{itemize}
        
        \item If $x< L/4$ and $t^{(1)}\in [L,5L/4)$, then $\alg$ has already served $[0,x)$ and is at position $L/2$ at $t=t^{(1)}$. Since $x < L/4$ and $t^{(1)} \geq L$, we get
        \begin{equation}
            |\opt| \geq \min\Bigl\{t^{(1)}+3L/4, t^{(1)}+L-x-(t^{(1)}-L) \Bigl\} 
        \label{eq:opt_4}
            = \min\Bigl\{t^{(1)}+3L/4, 2L-x \Bigl\}
            \geq 7L/4\enspace.
        \end{equation}

        \begin{itemize}
            \item If $\alg$ moves left at $t^{(1)}$, then $\alg$ has to serve $[x,L/4]$. So, using \eqref{eq:opt_4}, we get
            \begin{equation*}
                |\alg| \leq \max\Bigl\{t^{(1)} + L/2-x+L/4-x, |\opt| \Bigl\} + 3L/4 = |\opt|+3L/4 \enspace.
            \end{equation*}
            % From (\ref{eq:opt_4}) we get
            % \begin{equation}
            % \label{eq:alg_4a}
            %     |\alg| \leq \max\Bigl\{t^{(1)} + 3L/4-2x, |\opt| \Bigl\} + 3L/4 \leq |\opt| +3L/4 \enspace.
            % \end{equation}
            
            % From (\ref{eq:opt_4}), (\ref{eq:alg_4a}) it follows that
            Hence, we get
            $\ratio \leq 1 + \frac{3L/4}{ 7L/4} \leq 10/7.$
            % \begin{equation*}
            %     \ratio \leq 1 + \frac{3L/4}{ 7L/4} \leq 10/7\enspace.
            % \end{equation*}
            
            \item If $\alg$ moves right at $t^{(1)}$, then
                $|\alg| \leq \max\Bigl\{t^{(1)}+3L/4, |\opt|\Bigl\}+3L/4-x \\
 %               &= \max\Bigl\{t^{(1)}+3L/4-x, |\opt|-x\Bigl\}+3L/4 \\
                \leq |\opt|+3L/4$.
            Again, it follows that
            $\ratio  \leq 1+\frac{3L/4}{7L/4} \leq 10/7.$
            % \begin{equation}
            % \label{eq:alg_4b}
            %     |\alg| \leq |\opt|+3L/4 \enspace.
            % \end{equation}

            % \begin{equation*}
            %     \ratio  \leq 1+\frac{3L/4}{7L/4} \leq 10/7 \enspace.
            % \end{equation*}
        \end{itemize}
        \item If $x < L/4$ and $t^{(1)} \in [5L/4,7L/4)$, then $\opt$ has not served $[x,3L/4]$ yet. If $\opt$ does not turn around at any time, then $\opt \geq t^{(1)}+3L/4$. Otherwise, the first time that $\opt$ turns around is at $L$ and we have the following 2 cases. If $\opt$ goes straight to $L$, then $\opt \geq t^{(1)}+3L/4$. Otherwise, $\opt$ can be at $L/4$ having already served $[0,L/4)$ at time $t=L/2+x+L/4-x=3L/4$. Then, it reaches $L$ at $t=3L/4+3L/4=3L/2$ and it does not complete the tour before $t=3L/2+3L/4=9L/4$. So, we get
        % \begin{equation*}
        %     |\opt| \geq \min\Bigl\{t^{(1)}+3L/4-x,t^{(1)}+3L/4, 9L/4\Bigl\}
        % \end{equation*}
        % \begin{equation}
        % \label{eq:opt_5}
        %     = \min\Bigl\{t^{(1)}+3L/4-x, 9L/4\Bigl\} \enspace.
        % \end{equation}
        \begin{equation}
        \label{eq:opt_5}
            |\opt| \geq \min\Bigl\{t^{(1)}+3L/4-x,t^{(1)}+3L/4, 9L/4\Bigl\}
            = \min\Bigl\{t^{(1)}+3L/4-x, 9L/4\Bigl\} \enspace.
        \end{equation}
\allowdisplaybreaks
        \begin{itemize}
            \item If $\alg$ moves left at $t^{(1)}$, then it has to serve $[x,L/4]$.
            We have the following
            \begin{equation}
            \label{eq:alg_5}
                |\alg| \leq  \max\Bigl\{t^{(1)}+3L/4-2x, |\opt|\Bigl\}+3L/4 \enspace.
            \end{equation}
            From (\ref{eq:opt_5}), (\ref{eq:alg_5}) we have the following $4$ cases:
            \begin{align*}
            \text{(a) }&
                \ratio \leq \frac{t^{(1)}+3L/2-2x}{t^{(1)}+3L/4-x}
                \leq 1 + \frac{3L/4-x}{t^{(1)}+3L/4-x} \leq 10/7 \enspace,\\*
                &\text{since $x\leq L/4$ and $t^{(1)} \geq 5L/4$,}\\
            \text{(b) }&
                \ratio \leq \frac{|\opt|+3L/4}{|\opt|}
                \leq 1 + \frac{3L/4}{t^{(1)}+3L/4-x} \leq 10/7 \enspace,\\*
                &\text{since $x\leq L/4$ and $t^{(1)} \geq 5L/4$,}\\
            \text{(c) }&
                \ratio \leq \frac{t^{(1)}+3L/2-2x}{9L/4} 
                \leq \frac{t^{(1)}/L+3/2-2x/L}{9/4} \leq 13/9 \enspace,\\*
                &\text{since $x\leq L/4$ and $t^{(1)} \in [5L/4,7L/4)$,}\\
            \text{(d) }&
                \ratio \leq \frac{|\opt|+3L/4}{9L/4} \leq 4/3 \enspace.
            \end{align*}
            
            \item If $\alg$ moves right at $t^{(1)}$, then 
            \begin{equation}
            \label{eq:alg_6a}
                |\alg| \leq \max\Bigl\{t^{(1)} + 3L/4, |\opt|\Bigl\}+3L/4-x \enspace.
            \end{equation}
            From (\ref{eq:opt_5}), (\ref{eq:alg_6a}) we have the following $4$ cases:
            \begin{align*}
            \text{(a) }&
                \ratio \leq \frac{t^{(1)}+3L/2-x}{t^{(1)}+3L/4-x}
                \leq 1+ \frac{3L/4}{t^{(1)}+3L/4-x} \leq 10/7 \enspace,\\
                &\text{since $x\leq L/4$ and $t^{(1)} \geq 5L/4$,}\\
            \text{(b) }&
                \ratio \leq \frac{|\opt|+3L/4-x}{|\opt|}
                \leq 1+\frac{3L/4-x}{t^{(1)}+3L/4-x} \leq 10/7 \enspace,\\
                &\text{since $x\leq L/4$ and $t^{(1)} \geq 5L/4$,}\\
            \text{(c) }&
                \ratio \leq \frac{t^{(1)}+3L/2-x}{9L/4} \leq 13/9 \enspace,\\
                &\text{since $t^{(1)} < 7L/4$,}\\
            \text{(d) }&
                \ratio \leq \frac{|\opt|+3L/4-x}{9L/4} \leq 4/3 \enspace.
            \end{align*}
            
        \end{itemize}
        \item Else: $x < L/4$ and $t^{(1)} \geq 7L/4$, $|\alg|$ is at position $L/2$ at $t=t^{(1)}$ and we have that
        \begin{equation}
        \label{eq:opt_7}
           |\opt| \geq t^{(1)}+L/2 \enspace.
        \end{equation}
        We distinguish again the following cases:
        \begin{itemize}
            \item If $|\alg|$ moves left at $t^{(1)}$, then
            %\begin{equation}
            %\label{eq:alg_7a}
            $|\alg| \leq \max\Bigl\{t^{(1)} + \frac{3L}{4}-2x, |\opt|\Bigl\}+\frac{3L}{4}$. Together with \eqref{eq:opt_7}, we get
            %\end{equation}
            
            %From (\ref{eq:opt_7}), (\ref{eq:alg_7a}) we get
            \begin{align*}
                \ratio &\leq \frac{\max\Bigl\{t^{(1)} + 3L/4-2x, |\opt|\Bigl\}+3L/4}{|\opt|} \leq \max\Biggl\{ \frac{t^{(1)} + 3L/4}{t^{(1)}+L/2}, 1 \Biggl\} + \frac{3L/4}{t^{(1)}+L/2} \\
                &\leq \max\Biggl\{1+\frac{L/4}{t^{(1)}+L/2}, 1 \Biggl\} + 1/3 \leq \max\{ 10/9, 1\}+1/3 = 13/9\enspace,
            \end{align*}
            where we used $t^{(1)} \geq 7L/4$ for the third and forth inequality.
            % since $t^{(1)} \geq 7L/4$.
            
            % Using $x< L/4$ and $t^{(1)} \geq 7L/4$, we get
            % $\ratio \leq \max\{ 10/9, 1\}+1/3 = 13/9 .$
            % \begin{equation*}
            %     \ratio \leq \max\{ 10/9, 1\}+1/3 = 13/9 \enspace.
            % \end{equation*}
            
            \item If $\alg$ moves right at $t^{(1)}$, then we have that 
                $|\alg| \leq \max\Bigl\{t^{(1)} +\frac{3L}{4}, |\opt|\Bigl\}+\frac{3L}{4}-x$.
            Using this as well as (\ref{eq:opt_7}), we get
            \begin{align*}
                \ratio &\leq \frac{\max\Bigl\{t^{(1)} +3L/4, |\opt|\Bigl\}+3L/4-x}{|\opt|} \leq \max\Biggl\{ \frac{t^{(1)}+\frac{3L}{4}}{t^{(1)}+\frac{L}{2}}, 1 \Biggl\}
                + \frac{\frac{3L}{4}-x}{t^{(1)}+\frac{L}{2}}  \leq \frac{13}{9}\enspace.
            \end{align*}
        \end{itemize}
    \end{itemize}
    Consequently, we have an upper bound of $13/9$.
\end{proof}

We now focus on the closed variant of the problem and give a simple optimal algorithm that achieves a competitive ratio of $1$ (Algorithm~\ref{algo:closed_semi-line}). Basically, the algorithm moves to the furthest request and waits until the request is released. Then, it serves it and returns back to the origin, possibly waiting at the locations of unreleased requests, serving all remaining requests. Let us remark here that the semi-line is a special case of the star (with only a single branch) or the line, and our algorithm is essentially a special case of our Algorithm~\ref{algo:starClosedWaitFirst} for the star and of algorithm \textit{FARFIRST} for the line (in~\cite{GouleakisLS22}). The first condition in Algorithm~\ref{algo:starClosedWaitFirst} is always satisfied and gives us exactly the algorithm for the semi-line with its improved competitive ratio.

\begin{algorithm}[ht]
	\KwIn{Offline: Request locations $p_1, \dots, p_n$, where $p_1 \leq p_2 \leq ... \leq p_n$\\
	\phantom{\textbf{Input:} }Online: Release times $t_1, \dots, t_n$}
	Go to $p_n$.\\
	Wait until the request at $p_n$ is released and serve it. \\
	Move back towards the origin $O$, serving all unserved requests on the way, waiting for their release if necessary.
    \caption{Algorithm for closed OLTSP-L on the semi-line}
	\label{algo:closed_semi-line}
\end{algorithm}

\begin{proposition}\label{prop:semilineClosedUB}
	Algorithm \ref{algo:closed_semi-line} has a competitive ratio of $1$ for closed OLTSP-L on the semi-line.
\end{proposition}

\begin{proof}
In the closed case, $\opt$ has to serve the request at $L$ and also return to the origin. So, $|\opt| \geq 2L$.
If $\alg$ never waits for a request, then $|\alg|=2L$ and $\frac{|\alg|}{|\opt|} \leq 1$. Otherwise, let $t$ be the last time $\alg$ waits for a request and assume that $\alg$ is at distance $x$ from the origin. Then, $|\alg| = t + x$ and $\opt$ has to wait until $t$ to serve that request and also return to $O$. Thus, $|\opt| \geq t+x$ and we get $\frac{|\alg|}{|\opt|} \leq 1$.
\end{proof}

\section{Conclusion}\label{Conclusion}

In this work, we showed that the considered OLTSP-L problem admits an optimal $3/2$-competitive online algorithm for both variants on the general metric. Then, we provided several lower bounds and polynomial algorithms for some interesting metrics (ring, star, semi-line).

While the algorithm for the general metric is optimal, it does not run in polynomial time. A natural research direction would be to consider polynomial time algorithms for the general case. For instance, we can easily obtain a polynomial $2.5$-competitive algorithm (or more generally a $(1+r)$-competitive algorithm provided that there is an $r$-approximation algorithm for the off-line problem) for both variants of OLTSP-L as follows: Wait for all requests to be released and then take a tour/path. Clearly, $|\opt|$ is at least the time it takes for all requests to be released, hence if the chosen tour/path is optimal (resp. $r$-approximate) this algorithm has a competitive ratio of at most $2$ (resp. $1+r$). We get ratio 5/2 for the closed case by using   Christofides' heuristic~\cite{Christofides} to compute a TSP tour ($3/2$-approximation), and for the open case by using the $3/2$-approximation for path TSP in~\cite{Zenklusen}.

Note that the knowledge of the number of requests is sufficient to follow this strategy. Therefore, it is natural to consider if a better competitive ratio can be achieved with the extra knowledge of the locations of the requests. 

The study of other metrics could be also of interest, as well as improvements on the polynomial time upper bounds presented in this article.
	
Finally, we considered in this article only deterministic algorithms. Getting lower and upper bounds, both on general and specific metrics, using randomized algorithms is an interesting question.

\bibliography{bibliography}

\appendix
\section{Lower Bounds for a Known Number of Requests}\label{app:LBKnownNumber}

\subsection{Proof of Proposition~\ref{prop:KnownNumberClosedRing}}\label{app:KnownNumberClosedRing}

{\bf Proposition~\ref{prop:KnownNumberClosedRing}.} 
{\it 
When locations are unknown, for any $\epsilon > 0$, there is no $(2-\epsilon)$-competitive algorithm for closed OLTSP on the ring, even if the number of requests is known.}

\begin{proof}
    We adapt the proof of~\cite{AusielloFLST01} to use a known number of requests. For any $\epsilon > 0$ let $n = 6 \cdot \lceil 1/\epsilon \rceil + 1$. We consider a ring with circumference 1. Analogously to the previous notation, we denote each point on the ring by the time it takes a unit-speed server to reach this point when traversing the ring clockwise, starting from the origin. We distribute $4(n-1)/6$ requests equidistantly on the ring, i.e., for each $i = 1, \dots, 4(n-1)/6$, a request is placed at position $(i-1) \cdot 6/(4(n-1))$. Note that the distance between two requests is at most $\epsilon/4$, since 
	\[ \frac{1}{\frac{4(n-1)}{6}} = \frac{1}{4\left\lceil \frac{1}{\epsilon}\right\rceil} \le \frac{\epsilon}{4}
	\] holds. We release these requests with release time 0. Hence, $(n-1)/3+1$ many requests remain to be distributed. So far, the distributed requests can be served optimally by time $t=1$.
	
	Now, consider any online algorithm $\alg$. Here, we denote by $pos_{\alg}(t)$ the position of $\alg$ at time $t$. For some $\alpha \in [0,1/2]$, at time $1/2 + \alpha$, the online server is at one of the two points with distance $1/2-\alpha$ from the origin. To see that this holds, consider the following function $f\colon [0,1/2] \to [0,1/2]$, denoting the distance of the online server from the origin at time $1/2 + x$. Since $f$ is a continuous function, so is $g(x) = f(x)-x$, where $g(0) = f(0) \ge 0$ and $g(1/2) = f(1/2) - 1/2 \le 0$. Hence, there exists at least one $\alpha \in [0,1/2]$ with $g(\alpha) = 0$, or equivalently, $f(\alpha) = \alpha$. Let $\alpha$ be the smallest such value. Without loss of generality, we can assume that $pos_{\alg}(1/2 + \alpha) \le 1/2$, i.e., the online server visited all requests between the origin and $pos_{\alg}(1/2 + \alpha)$ on the clockwise tour. Hence, we assume without loss of generality that $pos_{\alg}(1/2+\alpha) = 1/2 - \alpha$. 
	
	The server might have gone a bit further and returned to $1/2-\alpha$, i.e., the server might have covered the segment $[1/2-\alpha,1/2-\alpha+y]$ for some $y \ge 0$. Similarly, the server might have visited parts of the counter-clockwise tour as well and returned to the origin before starting on the clockwise tour, i.e., the server might have covered a segment between $1-z$ and the origin for some $z \ge 0$. In the case of $z = 0$, we identify the point $1$ with the origin.
	Since the server had to cover both theses additional segments twice, their combined length is at most $y + z \le \alpha$, as the server is at a distance $1/2-\alpha$ at time $1/2 + \alpha$.
	
	Overall, the server still needs to visit the segment $(1/2-\alpha + y, 1-z)$ of size at least 2. Additionally, at $t = 1/2 + \alpha$, the remaining $(n-1)/3+1$ requests are revealed. They are distributed equidistantly on the segment $[0,1/2-\alpha]$, i.e., the already visited part. The distance between these requests is at most $\epsilon/4$, where the exact distance depends on $\alpha$.
	
	An optimal offline algorithm can still serve all requests in time $1$ by taking the counter-clockwise tour around the ring. To optimally finish the tour, the online algorithm can now either continue the clockwise tour and then visit the newly released requests (except for the one at $1/2-\alpha$, which can instantly be served). This leads to a completion time of at least $1/2 + \alpha + (1-1/2+\alpha) + 2\cdot (1/2-\alpha-\epsilon/4) = 2 - 1/2\epsilon > 2-\epsilon$. Alternatively, the online algorithm could also continue the clockwise tour up to the last unserved request (before reaching $1-z$) and turn back, going to the origin in a counter-clockwise fashion. This would lead to a completion time of at least $1/2 + \alpha + 2 \cdot (1/2+\alpha - z - \epsilon/4) + 1/2-\alpha = 2 + 2\alpha - 2z -\epsilon / 2 \ge 2 - \epsilon/2 > 2-\epsilon$ as $\alpha \ge z$.
	
	Since the offline optimum is 1, this means that no online algorithm can achieve a better competitive ratio than $2-\epsilon$, for any $\epsilon > 0$, even when the number of requests is known.
\end{proof}

\subsection{Proof of Proposition~\ref{prop:KnownNumberStar}}\label{app:KnownNumberStar}

{\bf Proposition~\ref{prop:KnownNumberStar}.} 
{\it 
When the locations are unknown, for any $\epsilon>0$ there is no $(2-\epsilon)$-competitive algorithm for both closed and open OLTSP on the star, even if the number of requests is known.}

\begin{proof}
Let us consider the closed version. Let $\epsilon>0$, $n=\lceil 7/\epsilon \rceil$. We consider a star with $k=\lfloor n/2 \rfloor$ branches, centered at $O$. Let us call $a_i$ the point on the $i$-st branch at distance 1 from $O$. All the locations will be either at the origin, or at some $a_i$. 

At time $t=1$ we release a first set of $k$ requests, one at each $a_i$. Consider an algorithm $\alg$ serving the requests. We proceed as follows until time $t^*=2k-1$: each time $\alg$ serves a request at some time $t$ at some point $a_i$, we release a new request at time $t+2$ on the same location $a_i$. 

Note that $\alg$ needs at least 2 units of time to serve a new request after having served one (it either has to move, or to wait 2 at some $a_i$). So at time $t^*$ it has served at most $k$ requests. Hence this way at most $k+k \leq n$ requests have been defined. If less than $n$ requests have been defined, we release at $t^*$ the remaining requests at $O$. 

At $t^*$, by construction there is still one request to serve in each branch for $\alg$, so it cannot finish before $t^*+2(k-1)=4k-3$. 

Now, let us show that $|\opt|\leq 2k+2$, which will lead to the result as $\frac{4k-3}{2k+2}=2-\frac{7}{2k+2}>2-\epsilon$.
For each branch $i$, let us consider $t^*_i$ the last time before (or at most) $t^*$ where $\alg$ went in $a_i$ ($t^*_i=0$ if $\alg$ did not visit $a_i$ before $t^*$). Up to relabeling, suppose that $t^*_1\leq t^*_2\leq \dots \leq t^*_k$. Note that if $t^*_i>0$ then $t^*_{i+1}\geq t^*_i+2$. As $t^*_k\leq t^*=2k-1$, we have $t^*_{k-1}\leq t^*-2=2(k-1)-1$ and, by an easy recurrence, $t^*_{k-i}\leq t^*-2i=2(k-i)-1$, or simply $t^*_i\leq 2i-1$ (which is also true for $t^*_i=0$). So the last request in branch $i$  occurs at time at most $2i-1+2=2i+1$. Then we can simply go to $a_i$ at time $2i+1$, for $i$ from 1 to $k$, and then go back to $O$. We serve all the requests in time $2k+2$.   
The same argument actually works for the open version as well.
\end{proof}

\subsection{Proof of Proposition~\ref{prop:KnownNumberClosedSemiline}}\label{app:KnownNumberClosedSemiline}

{\bf Proposition~\ref{prop:KnownNumberClosedSemiline}.} 
{\it 
When locations are unknown, for any $\epsilon > 0$, there is no $(4/3-\epsilon)$-competitive algorithm for closed OLTSP on the semi-line, even if the number of requests is known.}

\begin{proof}
Let $\epsilon > 0$ and a semi-line starting at the origin $O$ with length $1$. Consider an algorithm $\alg$ that serves the only one request and let $s$ be the distance of $\alg$ from $O$.

At time $t=1$:
\begin{itemize}
    \item if $s\geq 1/3$, then place and release a request at the origin $O$. In that case, $|\opt|=1$ and $|\alg| \geq 1+1/3=4/3$. So, $\frac{|\alg|}{|\opt|} \geq 4/3$.
    
    \item else: we place and release a request at position $1$. Then, $|\opt|=2$, $|\alg| \geq 2+2/3=8/3$ and $\frac{|\alg|}{|\opt|} \geq 4/3$.
\end{itemize}
Consequently, there is no online algorithm that can achieve a competitive ratio smaller than $4/3-\epsilon$, for any $\epsilon > 0$, even when the number of requests is known.
\end{proof}

\subsection{Proof of Proposition~\ref{prop:KnownNumberOpenSemiline}}\label{app:KnownNumberOpenSemiline}

{\bf Proposition~\ref{prop:KnownNumberOpenSemiline}.} 
{\it 
When locations are unknown, for any $\epsilon > 0$, there is no $(3/2-\epsilon)$-competitive algorithm for open OLTSP on the semi-line, even if the number of requests is known.}

\begin{proof}
Let $\epsilon > 0$ and a semi-line starting at the origin $O$ with length $1$. Consider an algorithm $\alg$ that serves the only one request and let $s$ be the distance of $\alg$ from $O$.

At time $t=1$:
\begin{itemize} 
    \item if $s\geq 1/2$, then place and release a request at the origin $O$. In that case, $|\opt|=1$ and $|\alg| \geq 1+1/2=3/2$. So, $\frac{|\alg|}{|\opt|} \geq 3/2$.
    
    \item else: we place and release a request at position $1$. Then, $|\opt|=1$, $|\alg| \geq 1+1/2$ and $\frac{|\alg|}{|\opt|} \geq 3/2$.
\end{itemize}
Therefore, $3/2$ is a lower bound for open OLTSP on the semi-line, even when the number of requests is known.
\end{proof}

\end{document}